\documentclass[a4paper,12pt]{article}

\usepackage{amssymb,color}
\usepackage{amsmath}
\usepackage[pdftex,bookmarksopen=true,bookmarks=true,unicode,setpagesize]{hyperref}
\hypersetup{colorlinks=true,linkcolor=black,citecolor=black}
\usepackage{amssymb, amsmath, amsthm}
\allowdisplaybreaks
\usepackage{cite}

\newtheorem{theorem}{Theorem}[section]
\newtheorem{corollary}[theorem]{Corollary}
\newtheorem{lemma}[theorem]{Lemma}
\newtheorem{proposition}[theorem]{Proposition}

\theoremstyle{remark}

\theoremstyle{remark}
\newtheorem{example}[theorem]{Example}
\theoremstyle{remark}
\newtheorem{remark}[theorem]{Remark}

\newcommand{\ran}{\operatorname{ran}}
\newcommand{\ls}{\operatorname{l.s.}}

\begin{document}

\begin{center}{\LARGE \bf
Fock representations of multicomponent (particularly non-Abelian anyon)  commutation relations}
\end{center}

{\large Alexei Daletskii}\\ Department of Mathematics, University of York, York YO1 5DD, UK;\\
e-mail: \texttt{alex.daletskii@york.ac.uk}\vspace{2mm}

{\large Alexander Kalyuzhny}\\ Institute of Mathematics, National Academy of Sciences of Ukraine, 3 Tereshenkivs'ka Str., 01601 Kyiv, Ukraine;
e-mail: \texttt{akalyuz@gmail.com}\vspace{2mm}

{\large Eugene Lytvynov}\\ Department of Mathematics,
Swansea University, Singleton Park, Swansea SA2 8PP, U.K.;
e-mail: \texttt{e.lytvynov@swansea.ac.uk}\vspace{2mm}

{\large  Daniil Proskurin}\\ Kyiv Taras Shevchenko University,
Faculty of Computer Science and Cybernetics, Volodymyrska 64,
01601 Kyiv, Ukraine;
e-mail: \texttt{prohor75@gmail.com}\vspace{2mm}

{\small
\begin{center}
{\bf Abstract}
\end{center}
\noindent Let $H$ be a separable Hilbert space and $T$ be a self-adjoint bounded  linear operator on $H^{\otimes 2}$ with norm $\le1$, satisfying  the Yang--Baxter equation. Bo\.zejko and Speicher (1994) proved that the operator $T$ determines a $T$-deformed Fock space $\mathcal F(H)=\bigoplus_{n=0}^\infty\mathcal F_n(H)$. We start with reviewing and extending the known results about the structure of the $n$-particle spaces $\mathcal F_n(H)$ and the commutation relations satisfied by the corresponding creation  and annihilation operators acting on $\mathcal F(H)$. 
We then choose $H=L^2(X\to V)$, the $L^2$-space of $V$-valued functions on $X$. Here $X:=\mathbb R^d$ and $V:=\mathbb C^m$ with $m\ge2$. Furthermore, we assume that the operator $T$ acting on $H^{\otimes 2}=L^2(X^2\to V^{\otimes 2})$ is given by $(Tf^{(2)})(x,y)=C_{x,y}f^{(2)}(y,x)$. Here, for a.a.\ $(x,y)\in X^2$, $C_{x,y}$ is a linear operator on $V^{\otimes 2}$ with norm $\le1$ that satisfies $C_{x,y}^*=C_{y,x}$ and the spectral quantum Yang--Baxter equation. The corresponding creation and annihilation operators describe a multicomponent quantum system. A special choice of the operator-valued function $C_{xy}$ in the case $d=2$ determines non-Abelian anyons (also called plektons). For a multicomponent system, we describe its $T$-deformed Fock space and the available commutation relations satisfied by the corresponding creation  and annihilation operators. Finally, we consider several examples of multicomponent quantum systems.
 } \vspace{2mm}

\vspace{2mm}


\section{Introduction}
This paper deals with the deformations of the canonical commutation/anticommutation relations that describe multicomponent quantum systems.

The first rigorous construction of a deformation
 of  the canonical (bosonic) commutation relations (CCR) and the canonical (fermionic) anticommutation relations (CAR) was given by Bo\.zejko and Speicher  \cite{BS}, see also Fivel~\cite{Fivel,Fivel2}, Greenberg~\cite{Greenberg}, Zagier~\cite{Zagier}. Let $H$ be a separable Hilbert space and let $q\in(-1,1)$.  On a $q$-deformed Fock space $\mathcal F(H)$ over $H$, Bo\.zejko and Speicher \cite{BS} constructed creation and annihilation operators $a^{+}(f)$  and  $a^{-}(f):=(a^+(f))^*$, respectively,  for $f\in  H$, that satisfy the $q$-commutation relations:
 \begin{equation}\label{iyt867}
a^-(f)a^+(g)=qa^+(g)a^-(f) + (f, g)_{H}, \quad f,g \in  H.
\end{equation}
Observe that the limiting values $q= 1$ and $q=-1$ correspond to  the CCR and CAR, respectively. In this case, one additionally has the {\it creation-creation} and {\it annihilation-annihilation} commutation relations
\begin{align}
a^+(f)a^+(g)&=qa^+(g)a^+(f),\notag\\
 a^-(f)a^-(g)&=qa^-(g)a^-(f) ,  \quad f,g \in  H,\ q=\pm1.\label{rsessaaa}\end{align}
respectively.

The operators $a^+(f)$, $a^-(f)$ ($f\in H$) from \cite{BS} form the Fock representation of the commutation relation \eqref{iyt867}. This means that there exists a vacuum vector $\Omega\in\mathcal F(H)$ that is cyclic for the operators $a^+(f)$ ($f\in H$) and satisfies
\begin{equation}\label{xrea56e}
a^-(f)\Omega=0\quad \text{for all }f\in H.\end{equation}
In fact, formulas \eqref{iyt867} and \eqref{xrea56e} and the condition of cyclicity of $\Omega$ uniquely identify the inner product on $\mathcal F(H)$. More precisely, the $q$-deformed Fock space has the form $\mathcal F(H)=\bigoplus_{n=0}^\infty\mathcal F_n(H)$ and the inner product on each $n$-particle space $\mathcal F_n(H)$ is determined by a bounded linear operator $\mathcal P_n$ on $H^{\otimes n}$, depending on $q$.
So one of the main achievements of \cite{BS} was the proof of the positivity of the operators $\mathcal P_n$ on $H^{\otimes n}$.  Unlike the case of CCR and CAR, for $q\in(-1,1)$ the kernel of $\mathcal P_n$ contains only zero, and so  $\mathcal F_n(H)$ coincides  as a set with $H^{\otimes n}$. This implies the absence of creation-creation and  annihilation-annihilation commutation relations, compare with \eqref{rsessaaa}. Note also that the creation and annihilation operators are bounded in the case $q\in[-1,1)$.

For studies of the $C^*$-algebras generated by the $q$-commutation relations, see e.g.\ \cite{dykemanicaJReineAngew1993,kennedynicaCMP2011,JSW}. The related von Neumann algebras  were studied e.g.\  in    \cite{nouMathAnn2004,sniadyCMP2004,ricardCMP2004,shlakhtenkoIMRN2004}.
The case $q=0$ corresponds to the creation and annihilation operators acting on the full Fock space; these operators are particularly important for models of free probability, see e.g.\ \cite{NiSp,bozejko-lytvynovCMP2009,BianeSpeicher}. Various aspects of noncommutative probability related to the general $q$-commutation relations \eqref{iyt867} were discussed e.g.\ in \cite{BS,bozejkokummererspeicherCMP1997,anshelevichDocMath2001,DM}.

An important generalization of the main result of \cite{BS} was obtained  in \cite{BozSpe}. Let $T$ be a self-adjoint bounded  linear operator on $H^{\otimes 2}$ with norm $\le1$, and assume that $T$ satisfies  the Yang--Baxter equation on  $H^{\otimes 3}$, see formula \eqref{braid} below. Then, similarly to the $q$~case, Bo\.zejko and Speicher \cite{BozSpe} defined a $T$-deformed Fock space $\mathcal F(H)=\bigoplus_{n=0}^\infty\mathcal F_n(H)$. To this end, they showed that, for each $n\in\mathbb N$, the corresponding operator $\mathcal P_n$ on $H^{\otimes n}$, depending on $T$, is positive. Furthermore,
 in the case $\|T\|<1$, the kernel of $\mathcal P_n$ contains only zero, and so  $\mathcal F_n(H)$ coincides  as a set with $H^{\otimes n}$.
 If the operator $T$ is given by
  $Tf\otimes g=qg\otimes f$ for $f,g\in H$,
 then one recovers  the $q$-deformed Fock space from \cite{BS}.

By using the $T$-deformed Fock space,  Bo\.zejko and  Speicher  \cite{BozSpe}
 constructed a Fock representation of the following discrete commutation relations between creation operators $\partial^\dag_i$ and annihilation operators $\partial_i$:
\begin{equation}\label{gtf67er}
\partial_i\partial_j^\dag=\sum_{k,l}T_{jl}^{ik}\,\partial^\dag_k\partial_l+ \delta_{i,j},\quad i,j\in\mathbb N.
\end{equation}
Here $(T_{ij}^{kl})_{i,j,k,l}$ is the matrix of the operator $T$ in a fixed orthonormal basis\footnote{Note, however, that the question of convergence of the series on the right-hand side of formula~\eqref{gtf67er} was not discussed in \cite{BozSpe}. So formula \eqref{gtf67er} was rigorously proved in \cite{BozSpe} only in the case where, for any fixed $i,j$, only a finite number of $T_{jl}^{ik}$ are not equal to zero}.
 In particular, for complex $q_{ij}$ with $\overline{q_{ij}}=q_{ji}$ and  $\sup_{i,j}|q_{ij}|\le 1$, one obtains the Fock representation of  the $q_{ij}$-commutation relations:
 \begin{equation}\label{ghdtrdeyk}
\partial_i\partial^\dag_j=q_{ij}\partial^\dag_j\partial_i+\delta_{ij},
\end{equation}
see also \cite{speicherLMPh1993}.

J\o rgensen, Schmitt and Werner \cite{jorgensenschmittwernerJFA1995} found sufficient conditions for the existence of the Fock representation of the commutation relations  \eqref{gtf67er} without requiring $T$ to satisfy the Yang--Baxter equation. 
For further results related to the  commutation relations~\eqref{gtf67er} or \eqref{ghdtrdeyk}, see  e.g.\   \cite{lustpiquard1999,Krolak, Krolak2,NelZeng,Meljanac Perica}. In the case  $\|T\|=1$,  J{\o}rgensen, Proskurin, and Samo\v{\i}lenko~\cite{JoProSa} found, for $n\ge2$, the kernel of the  operator  $\mathcal P_n$
that determines the inner product on  $\mathcal F_n(H)$.

Liguori and Mintchev \cite{LM} constructed the Fock representation of quantum fields with generalized statistics. Let $H=L^2(X)$,  the complex $L^2$-space on $X:=\mathbb R^d$.  Fix a function $Q:X^2\to\mathbb C$ satisfying $Q(x,y)=\overline{Q(y,x)}$ and $|Q(x,y)|=1$. Then the Fock representation of the corresponding generalized statistics is the family of the creation and annihilation operators on the $T$-deformed Fock space with the operator $T$ on $H^{\otimes 2}=L^2(X^2)$ given by
\begin{equation}\label{cxews5}
(Tf^{(2)})(x,y)=Q(x,y)f^{(2)}(y,x),\quad f^{(2)}\in H^{\otimes 2}.\end{equation}
Let us formally define creation operators $a^+(x)$ and annihilation operators $a^-(x)$ at points $x\in X$ that satisfy
$$a^+(f)=\int_X f(x)a^+(x)\,dx,\quad a^-(f)=\int_X \overline{f(x)}\,a^-(x)\,dx,\quad f\in H.$$
It is shown in \cite{LM} that these operators satisfy the $Q$-commutation relations
\begin{equation}
a^-(x)a^+(y)=Q(x,y)a^+(y)a^-(x)+\delta(x-y)\label{sseww}
\end{equation}
and
\begin{equation}
a^+(x)a^+(y)=Q(y,x)a^+(y)a^+(x),\quad a^-(x)a^-(y)=Q(y,x)a^-(y)a^-(x)\quad (x\ne y),\label{drsodoirtgu}
\end{equation}
the formulas making rigorous sense after smearing with a function $f(x)g(y)\in H^{\otimes 2}$. Note that, in this construction, the function $Q$ may be defined only  for a.a.\ $(x,y)\in X^2$.

In physics, generalized (intermediate) statistics have been discussed since
Leinass and Myr\-heim~\cite{Leinaas_Myrheim} conjectured their existence.   The first mathematically rigorous prediction of intermediate statistics was done by Goldin, Menikoff and Sharp \cite{GMS1,GMS2}. The name {\it anyon} was given to such statistics by Wilczek \cite{W1,W2}. Anyon statistics were used, in particular, to describe the quantum Hall effect, see e.g.\ \cite{Stern}.

Fix $q\in\mathbb C$ with $|q|=1$. Define a function $Q:X^2\to\mathbb C$ by
\begin{equation}\label{w5wdcftyw}
Q(x,y):=\begin{cases}q&\text{if }x^1<y^1,\\ \bar q,&\text{if }x^1>y^1,\end{cases}\end{equation}
where $x^1$ denotes the first coordinate of $x$. As shown by Goldin and Sharp \cite{Goldin_Sharp}, Goldin and Majid \cite{GoMa},  Liguori and Mintchev \cite{LM}, for $d=2$, the corresponding  commutation relations~\eqref{sseww}, \eqref{drsodoirtgu} describe anyons---particles associated with one-dimensional  unitary representations of the braid group.

Aspects of noncommutative probability related to anyons were discussed in  \cite{bozejkolytvynovwysoczanskiCMP2012,BLR}.
Lytvynov \cite{anyons} constructed a class of non-Fock representations of the anyon commutation relations for which the corresponding vacuum state is gauge-invariant quasi-free.

Note that, for any generalized statistics, the operator $T$ given by \eqref{cxews5} is unitary. In fact, for any operator $T$ that is additionally unitary, the corresponding operator $\mathcal P_n$ on $H^{\otimes n}$ is a multiple of an orthogonal projection.
See Bo\.zejko \cite{BBB} for a much weaker condition on $T$ that is sufficient for each operator $\mathcal P_n$ to be a multiple of an orthogonal projection.

Bo\.zejko, Lytvynov and Wysocza\'nski \cite{BLW-Q} discussed  Fock representations of the deformed commutation relations in the case where the operator $T$ is given by  formula \eqref{cxews5} in which the function $Q$ satisfies $Q(x,y)=\overline{Q(y,x)}$ and $|Q(x,y)|\le 1$. In this work, the $n$-particle subspaces $\mathcal F_n(H)$ were  described explicitly, and it was proved that the corresponding creation and annihilation operators satisfy the commutation relation \eqref{sseww}. Moreover, the creation-creation and annihilation-annihilation commutation relations \eqref{drsodoirtgu} hold for $x\ne y$ such that $|Q(x,y)|=1$:
\begin{align}a^+(x)a^+(y)&=Q(y,x)a^+(y)a^+(x),\notag\\
a^-(x)a^-(y)&=Q(y,x)a^-(y)a^-(x)\quad\text{if $x\ne y$ and $|Q(x,y)|=1$}.\label{des53wrtgu}
\end{align}

In the present paper, by a multicomponent quantum system we  understand a family of creation and annihilation operators  $a^+(f)$, $a^-(f)$ on a $T$-deformed Fock space $\mathcal F(H)$, where $f$ belongs to $H=L^2(X\to V)$, the $L^2$-space of $V$-valued functions on $X$. Here $V:=\mathbb C^m$ with $m\ge2$. Furthermore, we assume that the operator $T$ acting on $H^{\otimes 2}=L^2(X^2\to V^{\otimes 2})$ is given by
\begin{equation}\label{T-C-map}
(Tf^{(2)})(x,y):=C_{x,y}f^{(2)}(y,x),\quad f^{(2)}\in H^{\otimes 2}.
\end{equation}
Here $C_{x,y}$ is a linear operator on $V^{\otimes 2}$  with norm $\le1$, which is defined for a.a.\ $(x,y)\in X^2$ and  satisfies the symmetry relation $C_{x,y}^*=C_{y,x}$ together with the spectral quantum Yang--Baxter equation, see formula~\eqref{YBxyz} below.
Under the assumption that, for a.a.\ $(x,y)\in X^2$, $C_{x,y}$ is a unitary operator on $V^{\otimes 2}$ (or, equivalently, $T$ is a unitary operator on $H^{\otimes 2}$), the multicomponent quantum systems were discussed in \cite{LM}, see also the references therein. 

A multicomponent counterpart of  an anyon system was originally called  {\it plektons}, see e.g.\ \cite{GoMa}. The first publication pointing out the possibility of such a quantum system was the  comment by Menikoff, Sharp, and Goldin \cite{GoldinMenikoffSharp}.  Plektons are  quasiparticles in dimension $d=2$ that are associated with higher-dimensional (non-Abelian) unitary representations of the braid group. In view of this, more recently these quasiparticles have been mostly called {\it non-Abelian anyons}, the term that will be used in the present paper. Non-Abelian anyons form a central tool in topological quantum computation, see e.g.\ \cite{Pachos,Stanescu}.

According to \cite{GoMa}, a non-Abelian anyon system is determined by a unitary operator $C$ on $V^{\otimes 2}$, which defines $C_{x,y}$ in \eqref{T-C-map}  via the formula
\begin{equation}\label{sr4q4w2}
C_{xy}:=\begin{cases}C,& \text{if }x^1<y^1, \\
C^*,&\text{if } x^1>y^1,\end{cases}\end{equation}
compare with \eqref{w5wdcftyw}. The operator $T$ satisfies the Yang--Baxter equation on $H^{\otimes 3}$ if and only if the operator $C$ satisfies the Yang--Baxter equation on $V^{\otimes 3}$, see Lemma~\ref{drs5ywuy} below. In the latter case, the operator $C$ determines, for each $n\ge2$, a (non-Abelian) unitary representation of the braid group $B_n$.

The paper is organized as follows. In Section~\ref{esa4w4qpp}, we review and extend the results of \cite{BozSpe,{JoProSa}} regarding the general deformed commutation relations governed by a bounded linear operator $T$ satisfying the assumptions of the paper \cite{BozSpe}. Our man results in this section are as follows.

\begin{itemize}
\item[(i)] In the case  $\|T\|=1$, we clarify the  structure of the $n$-particle subspaces $\mathcal F_n(H)$ of the $T$-deformed Fock space $\mathcal F(H)$ (Theorem~\ref{se5sw5yw5wq} and Corollary~\ref{ysd6qwe}). Furthermore, we show that the orthogonal projection $\mathbb P_n$ of $H^{\otimes n}$ onto its subspace $\mathcal F_n(H)$ can be represented, for $n\ge3$, as (a multiple of) the parallel sum of two explicitly given orthogonal projections, built with the help of $\mathbb P_2$  (Proposition~\ref{rw5uwu83xx}).

\item[(ii)] We find all possible commutation relations between the operators $a^{\pm}(f)$ and $a^{\pm}(g)$ (Theorem~\ref{rdw6uged}).
\end{itemize}

Note that previously the commutation relations between two creation operators and between two annihilation operators have only been found in the case where the operator $T$ is given by formula~\eqref{cxews5}, see \cite{BLW-Q,GoMa,LM}.

In Section~\ref{e4q4gdwyfd}, we consider the general multicomponent quantum systems. We apply the results of Section~\ref{esa4w4qpp} to the case where the operator $T$ is given by formula~\eqref{T-C-map}. The main results of this section---Theorems~\ref{vytr7i5}, \ref{yjsqrd6uq} and Corollaries~\ref{formalCR}, \ref{formalCRU}---describe the corresponding $T$-deformed Fock space and the available commutation relations between the creation/annihilation operators. In particular, we find a multicomponent counterpart of the commutation relations \eqref{sseww}, \eqref{des53wrtgu}.

Finally, in Section~\ref{sec-examp}, we consider several examples of multicomponent quantum systems. These include examples when the operator-valued function $C_{x,y}$ in formula~\eqref{T-C-map} is constant, i.e., $C_{x,y}=C$ for all $x,y$, examples of non-Abelian anyon quantum systems and other.
In these examples, we give explicit description of the corresponding Fock space $\mathcal F(H)$ and the orthogonal projection $\mathbb P_2$ of $H^{\otimes n}$ onto  $\mathcal F_2(H)$, and calculate the available commutation relations.

\section{General $T$-deformed commutation relations}\label{esa4w4qpp}

\subsection{$T$-deformed tensor power of a Hilbert space}\label{es5w72dc}

For a Hilbert space $\mathcal H$, let $\mathcal L(\mathcal H)$ denote the space of all bounded linear operators on $\mathcal H$.  We will denote by $\mathbf 1_{\mathcal H}$ the identity operator on $\mathcal H$. However, where the Hilbert space in consideration is clear from the context, we will just use $\mathbf 1$ for the identity operator on this space.

Let $H$ be a separable  complex Hilbert space, and let  $T\in\mathcal L(H^{\otimes 2})$.  We assume that $T$ is self-adjoint,  $\|T\|\le 1$, and $T$ satisfies the Yang--Baxter equation on  $H^{\otimes 3}$:
\begin{equation}
(T\otimes \mathbf 1_H)(\mathbf 1_H\otimes T)(T\otimes \mathbf 1_H)=(\mathbf 1_H\otimes T)(T\otimes \mathbf 1_H)(\mathbf 1_H\otimes T).
\label{braid}
\end{equation}%
  For $i\in\mathbb N$, we denote by $T_i$ the operator on $H^{\otimes n}$ with $n\ge i+1$ given by
  $$T_i:=\mathbf 1_{H^{\otimes(i-1)}}\otimes T\otimes\mathbf 1_{H^{\otimes(n-i-1)}}.$$

Let $S_n$ denote the symmetric group of degree $n$. We represent  a permutation $\sigma\in S_n$ as an arbitrary  product of adjacent transpositions,
\begin{equation}\label{representation of permutation}
\sigma = \sigma_{j_1}\dotsm \sigma_{j_m},
\end{equation}
 where
$\sigma_j:=(j, j+1)\in S_n$ for $1\leq j \leq n-1$.
 A permutation $\sigma
\in S_n$ can be represented (not in a unique way, in general) as a
reduced product of a minimal number of adjacent  transpositions,  i.e., in the form  \eqref{representation of permutation} with a minimal $m$. Then the mapping $
\sigma_k\mapsto T_{\sigma_k}:=T_k\in\mathcal L(H^{\otimes n})
$ can be multiplicatively extended
to $S_n$ by setting
\begin{equation}\label{operators Psi_pi}
S_n\ni\sigma \mapsto T_{\sigma}:=T_{j_1}\dotsm T_{j_m}.
\end{equation}
(For the identity permutation $e\in S_n$, $T_e:=\mathbf 1$.)
Although
representation (\ref{representation of permutation}) of $\sigma\in S_n$ in a reduced form is not unique,
formula (\ref{braid}) implies that the extension
(\ref{operators Psi_pi}) is well-defined, i.e., it does not depend on
the representation of $\sigma$.

For each $n\ge 2$, we define  $\mathcal P_n\in\mathcal L(H^{\otimes n})$ by
$$\mathcal P_n:=\sum_{\sigma\in S_n}T_\sigma.$$ By \cite{BozSpe}, the operator $\mathcal P_n$ is positive, i.e., $\mathcal P_n\ge0$, and in the case $\|T\|<1$, it is strictly positive. We denote
\begin{equation}\label{vtrw6u4}
\mathcal F_n(H):=\ker(\mathcal P_n)^\perp=\overline{\operatorname{ran}(\mathcal P_n)},\end{equation}
 i.e., the orthogonal compliment of the kernel of $\mathcal P_n$ in $H^{\otimes n}$, or equivalently the closure of the range of $\mathcal P_n$. As easily seen, the operator $\mathcal P_n$ is strictly positive on $\mathcal F_n(H)$, so one can introduce a new inner product on $\mathcal F_n(H)$ by
$$(f^{(n)},g^{(n)})_{\mathcal F_n(H)}:=(\mathcal P_nf^{(n)},g^{(n)})_{H^{\otimes n}},\quad f^{(n)}, g^{(n)}\in \mathcal F_n(H),$$
which makes $\mathcal F_n(H)$ a Hilbert space. Note that, if $T=\mathbf 0$, the Hilbert spaces $\mathcal F_n(H)$ and $H^{\otimes n}$ coincide.  Thus, a non-zero operator $T$ leads to a deformation  of the Hilbert space $H^{\otimes n}$.

Let $\mathbb P_n$ denote the orthogonal projection of the Hilbert space $H^{\otimes n}$ onto its subspace $\mathcal F_n(H)$.

Assume in addition that the operator $T$ is unitary. Then mapping \eqref{operators Psi_pi} determines a unitary representation of $S_n$, hence in formula \eqref{operators Psi_pi} $\sigma$ should not necessarily be in a reduced form.  This implies the equality $\mathbb P_n=\frac1{n!}\mathcal P_n$, which does not hold in the general case.

As already mentioned before, in the case $\|T\|<1$, $H^{\otimes n}$ and $\mathcal F_n(H)$ coincide as sets. In the case where $\|T\|=1$, the following result shown in \cite{JoProSa} gives a description of $\ker(\mathcal P_n)=\ker(\mathbb P_n)$:
\begin{equation}
\ker(\mathcal{P}_{n})=\overline{\sum_{i=1}^{n-1}\ker (\mathbf{1}+T_{i})},
\label{KerPn}
\end{equation}%
i.e., the kernel of $\mathcal{P}_{n}$ is equal to the closure of the linear
span of the subspaces $\ker(\mathbf{1}+T_{i})$, $i=1,\dots ,n-1$.
Note that formula \eqref{KerPn} remains true when $\|T\|<1$, in which case it gives $\ker(\mathcal{P}_{n})=\{0\}$.
Since $\ker(\mathbf{1}+T_{i})^\perp=\overline{\operatorname{ran}(\mathbf{1}+T_{i})}$, formulas \eqref{vtrw6u4} and \eqref{KerPn} imply
\begin{equation}\label{dr6eu48}
\mathcal F_n(H)=\bigcap_{i=1}^{n-1}\overline{ \operatorname{ran}(\mathbf{1}+T_{i})}.\end{equation}

In view of formula \eqref{dr6eu48}, we can now give a representation of the orthogonal projection $\mathbb P_n$ onto $\mathcal F_n(H)$ by using the notion of a parallel sum of two projections, see e.g.\ \cite{morley}.  Let us first recall this notion.  Let $\mathcal H$ be a complex Hilbert space, let $\mathcal H_1$ and $\mathcal H_2$ be closed subspaces of $\mathcal H$, and let $P_1$ and $P_2$ denote the orthogonal projections of $\mathcal H$ onto $\mathcal H_1$ and $\mathcal H_2$, respectively. The parallel sum of $P_1$ and $P_2$, denoted by $(P_1:P_2)$, is the self-adjoint bounded linear operator on $\mathcal H$ defined by its quadratic form
\begin{equation}\label{re6ewfe}
((P_1:P_2)x,x)_{\mathcal H}=\inf_{y+z=x}\big((P_1y,y)_{\mathcal H}+(P_2z,z)_{\mathcal H}\big),\quad x\in\mathcal H.\end{equation}
The right-hand side of \eqref{re6ewfe} is equal to $\frac12\|x\|_{\mathcal H}^2$
for $x\in\mathcal H_1\cap\mathcal H_2$ and equal to zero for $x\in(\mathcal H_1\cap\mathcal H_2)^\perp$. Hence, $2(P_1:P_2)$ is the orthogonal projection of $\mathcal H$ onto $\mathcal H_1\cap\mathcal H_2$. Observe that $2(P_1:P_2)=P_1P_2$ if and only if $P_1$ and $P_2$ commute, or equivalently $\mathcal H_1\perp\mathcal H_2$.

Denote $\mathfrak P:=\mathbb P_2$ and analogously to operators $T_i$ define operators $\mathfrak P_i$. Then, for $n\ge i+1$,  $\mathfrak P_i$ is the orthogonal projection of $H^{\otimes n}$ onto $\overline{\operatorname{ran}(\mathbf{1}+T_{i})}$.

\begin{proposition}\label{rw5uwu83xx}
Let  $n\ge 3$. Define operators $Q_1$ and $Q_2$ on $H^{\otimes n}$ by
$$Q_{1}:=\prod\limits_{\substack{ i\leq n-1  \\ \text{\rm $i$ odd}}}\mathfrak{P}_{i},\quad Q_{2}:=\prod\limits_{\substack{ i\leq n-1  \\ \text{\rm $i$ even}}}\mathfrak{P}_{i}.$$
Then the operators $Q_1$ and $Q_2$ are orthogonal projections and $\mathbb P_n=2(Q_1:Q_2)$.  Furthermore, for each $m\in\mathbb N$, $2\le m\le n-1$,
\begin{equation}\label{vydr6ei6}
\mathbb P_{n}=\mathbb P_{n}(\mathbb P_m\otimes \mathbf 1_{H^{\otimes(m-n)}})=\mathbb P_{n}(\mathbf 1_{H^{\otimes(m-n)}}\otimes\mathbb P_m).
\end{equation}
\end{proposition}

\begin{proof}
Observe that the projections $\mathfrak P_i$ with odd (respectively even) $i$ mutually commute. This implies that $Q_1$ and $Q_2$ are orthogonal projections onto
$$\bigcap_{\substack{ i\leq n-1  \\ \text{$i$ odd}}}  \overline{\operatorname{ran}(\mathbf{1}+T_{i})}\quad  \text{and}\quad\bigcap_{\substack{ i\leq n-1  \\ \text{$i$ even}}}  \overline{\operatorname{ran}(\mathbf{1}+T_{i})},$$
respectively. Formula \eqref{dr6eu48} implies $\mathbb P_n=2(Q_1:Q_2)$.

Let us prove the first equality in \eqref{vydr6ei6}, the second one being proved similarly. The operator $\mathbb P_m\otimes \mathbf 1_{H^{\otimes (m-n)}}$ is the orthogonal projection of $H^{\otimes n}$ onto the subspace
$$\bigcap_{i=1}^{m-1}\overline{ \operatorname{ran}(\mathbf{1}+T_{i})}.$$ But $\mathcal F_n(H)$ is a subspace of this space (see \eqref{re6ewfe}), which implies the statement.
\end{proof}

\begin{theorem}\label{se5sw5yw5wq}
For each $n\ge2$, we have
\begin{equation}\label{cx4aq4}
\mathcal F_n(H)=\big\{f^{(n)}\in H^{\otimes n}\mid (\mathbf 1-T_i)f^{(n)}\in \overline{\operatorname{ran}(\mathbf 1-T_i^2)},\ i=1,2,\dots,n-1\big\}. \end{equation}
\end{theorem}

\begin{proof} Note that, when $\|T\|<1$, formula \eqref{cx4aq4} just states the known equality $\mathcal F_n(H)=H^{\otimes n}$. So we only need to prove formula \eqref{cx4aq4} in the case $\|T\|=1$.
We start with the following lemma.

\begin{lemma}\label{ker-gen}
The kernel of the operator $\mathbf{1}+T$ has the
following representation:
\begin{equation}
\ker(\mathbf 1+T)=(\mathbf 1-T)\ker(\mathbf 1-T^2).
\label{KerrT}\end{equation}
\end{lemma}

\begin{proof}
If $f^{(2)}\in \ker(\mathbf 1-T^2)$, then $(\mathbf 1-T)f^{(2)}\in\ker(\mathbf 1+T)$, which implies the inclusion
$$(\mathbf 1-T)\ker(\mathbf 1-T^2)\subset \ker(\mathbf 1+T).$$
To prove the converse inclusion, take any $f^{(2)}\in \ker(\mathbf 1+T)$ (i.e., $f^{(2)}=-Tf^{(2)}$). Then
$\frac12\,f^{(2)}\in  \ker(\mathbf 1-T^2)$ and
$$(\mathbf 1-T){\textstyle\frac12}\,f^{(2)} ={\textstyle\frac12}\,f^{(2)}+{\textstyle\frac12}\,f^{(2)}=f^{(2)}.$$
Thus, formula \eqref{KerrT} is shown.
\end{proof}

In the case $n=2$, formula \eqref{cx4aq4} states
\begin{equation}\label{tesw5a5e}
\mathcal F_2(H)=\big\{f^{(2)}\in H^{\otimes 2}\mid (\mathbf 1-T)f^{(2)}\in \overline{\operatorname{ran}(\mathbf 1-T^2)}\big\}.\end{equation}
Let us now prove this formula. Observe that
\begin{equation}\label{orthdec}
H^{\otimes 2}=\ker(\mathbf 1+T)\oplus
\ker(\mathbf 1-T)\oplus \overline{\operatorname{ran}(\mathbf 1-T^2)}.
\end{equation}
Thus, each  $f^{(2)}\in H^{\otimes 2}$ can be represented as
 $f^{(2)}=f^{(2)}_1+f^{(2)}_2+f^{(2)}_3$, where $f^{(2)}_1\in \ker(\mathbf 1+T)$, $f^{(2)}_2\in \ker(\mathbf 1-T)$, $f^{(2)}_3\in \overline{\operatorname{ran}(\mathbf 1-T^2)}$, and $f^{(2)}\in \mathcal F_2(H)$ if and only if $f^{(2)}_1=0$. Note that the subspaces $\ker(\mathbf 1+T)$, $\ker(\mathbf 1-T)$, and $\overline{\operatorname{ran}(\mathbf 1-T^2)}$ are invariant for the operator $T$, hence also for the operator $\mathbf 1-T$. Therefore, for $f^{(2)}\in H^{\otimes 2}$, we get
 $$(\mathbf 1-T)f^{(2)}=(\mathbf 1-T)f_1^{(2)}+(\mathbf 1-T)f_3^{(2)}$$
 with $(\mathbf 1-T)f_1^{(2)}\in\ker(\mathbf 1-T)$ and $(\mathbf 1-T)f_3^{(2)}\in \overline{\operatorname{ran}(\mathbf 1-T^2)}$. Hence, condition
 \begin{equation}\label{xerwq5u636ui}
(\mathbf 1-T)f^{(2)}\in \overline{\operatorname{ran}(\mathbf 1-T^2)}\end{equation}
  is satisfied if and only if $(\mathbf 1-T)f_1^{(2)}=0$. But since $f_1^{(2)}\in \ker(\mathbf 1+T)$,  we have $(\mathbf 1-T)f_1^{(2)}=0$ if and only if $f^{(2)}_1=0$. Thus, formula~\eqref{tesw5a5e} is proved.

 For $n\ge3$, formula~\eqref{cx4aq4} follows from \eqref{dr6eu48} and~\eqref{tesw5a5e}.
\end{proof}

\begin{corollary}\label{ysd6qwe} Assume additionally that the operator $T$ is unitary. Then, for each $n\ge2$,
$$\mathcal F_n(H)=\big\{f^{(n)}\in H^{\otimes n}\mid T_if^{(n)}=f^{(n)},\ i=1,2,\dots,n-1\big\}. $$
\end{corollary}

\begin{proof}
Since $T$ is both self-adjoint and unitary, we have $T^{-1}=T$. Hence,
$$\operatorname{ran}(\mathbf 1-T^2)=\{0\}.$$
Now the statement follows from Theorem~\ref{se5sw5yw5wq}.
\end{proof}

\begin{remark}
In view of Corollary \ref{ysd6qwe}, in the case where $T$ is additionally unitary, we can interpret $\mathcal F_n(H)$ as the $n$th $T$-symmetric tensor power of $H$.
\end{remark}

\subsection{Creation and annihilation operators on the full Fock space}\label{utkfwedt7}

We will now recall and extend the construction of creation and annihilation operators acting on the full Fock space, compare with e.g.\ \cite[Lecture 7]{NiSp}.

Let $H_{\mathbb R}$ be a real separable Hilbert space, and let  $H$ denote the complex Hilbert space constructed as the complexification of $H_{\mathbb R}$. More precisely, elements of $H$ are of the form $f_1+if_2$ with $f_1,f_2\in H_{\mathbb R}$. For $f=f_1+if_2,\, g=g_1+ig_2\in H$, we denote
$$\langle f,g\rangle:=(f_1,f_2)_{H_{\mathbb R}}-(g_1,g_2)_{H_{\mathbb R}}+i\big((f_1,g_2)_{H_{\mathbb R}}+(f_2,g_1)_{H_{\mathbb R}}\big),$$
i.e., $\langle\cdot,\cdot\rangle$ is the extension of the inner product on $H_{\mathbb R}$ by linearity to $H\times H$. Then the inner product on $H$ is given by
$(f,g)_H:=\langle f,Jg\rangle$,
where
\begin{equation}\label{compconj}
Jg=J(g_1+ig_2):=g_1-ig_2
\end{equation}
is the complex conjugation on the space $H$ considered as the complexification of $H_{\mathbb R}$.

Let
$\Gamma(H)$ denote the full Fock space over $H$:
$$\Gamma(H):=\bigoplus_{n=0}^\infty H^{\otimes n}.$$
Here $H^0:=\mathbb C$. The vector $\Omega:=(1,0,0,0,\dots)\in\Gamma(H)$ is called the vacuum.

For each $f\in H$, we denote by $l^+(f)$ the  operator of left creation by $f$. This is the bounded linear operator on $\Gamma(H)$ satisfying
$l^+(f)\Omega=f$ and
$l^+(f) g^{(n)}=f\otimes g^{(n)}$ for $g^{(n)}\in H^{\otimes n}$, $n\in\mathbb N$. Note that $\|l^+(f)\|=\|f\|$.  The operator of left annihilation by $f$ is defined by
$$l^-(f):=l^+(Jf)^*.$$
This operator satisfies
\begin{align*}
&l^-(f)\Omega=0,\\
&l^-(f)g_1\otimes g_2\otimes\dots\otimes g_n=
\langle f,g_1\rangle g_2\otimes\dots\otimes g_n,\quad g_1,g_2,\dots,g_n\in H,\ n\in\mathbb N.\end{align*}

For $f_1,f_2,g\in H$, we denote
$\langle f_1\otimes f_2,g\rangle_2:=\langle f_2,g\rangle f_1$.
As easily seen, $\langle f_1\otimes f_2,\cdot\rangle_2$ determines a Hilbert--Schmidt operator on $H$. Extending this definition by linearity and continuity, we define, for an arbitrary $f^{(2)}\in H^{\otimes 2}$, a Hilbert--Schmidt operator  $\langle f^{(2)},\cdot\rangle_2$ on $H$ with Hilbert--Schmidt norm
$\|f^{(2)}\|_{H^{\otimes 2}}$.

For $f_1,f_2\in H$ and $g^{(n)}\in H^{\otimes n}$, we have
\begin{equation}\label{tsesw5w5}
l^+(f_1)l^-(f_2)g^{(n)}=\big(\langle f_1\otimes f_2,\cdot\rangle _2\big)\otimes\mathbf 1_{H^{\otimes(n-1)}}g^{(n)}.\end{equation}
Indeed, choosing $g^{(n)}=g_1\otimes g_2\otimes\dots\otimes g_n$ with $g_1,\dots,g_n\in H$, we get
$$l^+(f_1)l^-(f_2)g^{(n)}=\langle f_2,g_1\rangle f_1\otimes g_2\otimes\dots\otimes g_n=\big(\langle f_1\otimes f_2,g_1\rangle _2\big)\otimes g_2\otimes\dots\otimes g_n,$$
which implies \eqref{tsesw5w5}. 
In view of \eqref{tsesw5w5}, for each $f^{(2)}\in H^{\otimes 2}$, we can define a bounded linear operator $l^{+-}(f^{(2)})$ on $\Gamma(H)$ by
\begin{align*}
&l^{+-}(f^{(2)})\Omega:=0,\\
&l^{+-}(f^{(2)})g^{(n)}:= \big(\langle f^{(2)},\cdot\rangle _2\big)\otimes\mathbf 1_{H^{\otimes(n-1)}}\,,\quad g^{(n)}\in H^{\otimes n},\ n\in\mathbb N.
\end{align*}

Let $(e_i)_{i\ge1}$ be an orthonormal basis of $H_{\mathbb R}$, hence also an orthonormal basis of $H$. Then, for each $f^{(2)}\in H^{\otimes 2}$, we easily see that
$$l^{+-}(f^{(2)})=\sum_{i,j}\langle f^{(2)},e_i\otimes e_j\rangle l^+(e_i)l^-(e_j),$$
where the series converges in the operator norm. Here and below, for  $f^{(2)},g^{(2)}\in H^{\otimes 2}$, we use the notation $\langle f^{(2)},g^{(2)}\rangle:=(f^{(2)},Jg^{(2)})_{H^{\otimes 2}}$, where $J$  denotes the complex conjugation on  $H^{\otimes 2}$, cf. \eqref{compconj}.

Similarly, for each $f^{(2)}\in H^{\otimes 2}$, we define bonded linear operators $l^{++}(f^{(2)})$ and $l^{--}(f^{(2)})$ on $\Gamma(H)$ that satisfy
\begin{align*}
l^{++}(f^{(2)})&=\sum_{i,j}\langle f^{(2)},e_i\otimes e_j\rangle l^+(e_i)l^+(e_j),\\
l^{--}(f^{(2)})&=\sum_{i,j}\langle f^{(2)},e_i\otimes e_j\rangle l^-(e_i)l^-(e_j),
\end{align*}
the series converging in the operator norm.
Note that $l^{++}(f^{(2)})^*=l^{--}(\mathbb Sf^{(2)})$, where $\mathbb S$ denotes the continuous antilinear operator on $H^{\otimes 2}$ satisfying
$$\mathbb Sf\otimes g:=J(g\otimes f),\quad f,g\in H.$$

\subsection{Creation and annihilation operators on the $T$-deformed Fock space}\label{uerqdqdf}

Let an  operator $T$ and a Hilbert space $H$ be as in Subsection \ref{es5w72dc} and \ref{utkfwedt7}, respectively.
We define the $T$-deformed Fock space over $H$ by
$$\mathcal F (H):=\bigoplus_{n=0}^\infty \mathcal F_n(H).$$
Here $\mathcal F_0(H):=\mathbb C$ and the vector $\Omega=(1,0,0,0,\dots)$ is still called the vacuum.
Note that the full Fock space $\Gamma(H)$ is the special case of $\mathcal F (H)$ for $T=\mathbf 0$.

Let $\mathcal F_{\mathrm{fin}} (H)$ denote the subspace of $\mathcal F (H)$ that consists of all finite sequences $f=(f^{(0)},f^{(1)},\dots,f^{(n)},0,0,\dots)$ with $f^{(i)}\in\mathcal F_i(H)$. We equip $\mathcal F_{\mathrm{fin}} (H)$ with the topology of the topological direct sum of the $\mathcal F^{(n)}(H)$ spaces. Thus, convergence of a sequence in $\mathcal F_{\mathrm{fin}} (H)$ means uniform finiteness and coordinate-wise convergence of non-zero coordinates. We denote by $\mathcal L(\mathcal F_{\mathrm{fin}} (H))$ the space of all continuous linear operators on $\mathcal F_{\mathrm{fin}} (H)$.

For $f\in H$, we define a creation operator $a^+(f)$
as the linear operator on $\mathcal F_{\mathrm{fin}} (H)$ given by
\begin{align*}
&a^+(f)\Omega:=f,\\
&a^+(f)g^{(n)}:=\mathbb P_{n+1}l^+(f)g^{(n)}=\mathbb P_{n+1}(f\otimes g^{(n)}),\quad g^{(n)}\in\mathcal F_n(H),\ n\in\mathbb N.
\end{align*}
Note that formula \eqref{vydr6ei6} implies that
\begin{equation}\label{yre6ueu}
a^+(f)\mathbb P_ng^{(n)}=\mathbb P_{n+1}l^+(f) g^{(n)},\quad g^{(n)}\in H^{\otimes n},\ n\ge 2.\end{equation}

Next, for $f\in H$, we define an annihilation operator $a^-(f)$
on $\mathcal F_{\mathrm{fin}} (H)$ by
$$a^-(f):=a^+(Jf)^*\restriction\mathcal F_{\mathrm{fin}} (H).$$
By \cite{BozSpe}, one has the following explicit formula for the action of $a^-(f)$:
\begin{equation}\label{yjqdse6ytqsd}
a^-(f)g^{(n)}=\mathbb P_{n-1}l^-(f)\mathbb T_n
g^{(n)},\quad g^{(n)}\in\mathcal F_n(H),\end{equation}
where $\mathbb T_n\in\mathcal L(H^{\otimes n})$ is defined by
\begin{equation}\label{yfdqfqdr}
\mathbb T_n:=\mathbf 1+T_1+T_1T_2+\dots+T_{1}T_2\dotsm T_{n-1}.\end{equation}

\begin{proposition}\label{yedri751r}
For each $f\in H$, $a^+(f), a^-(f)\in \mathcal L(\mathcal F_{\mathrm{fin}} (H))$.
\end{proposition}

\begin{proof}
It is sufficient to prove that, for each $n\in\mathbb N$, the operators $a^+(f):\mathcal F_n(H)\to\mathcal F_{n+1}(H)$ and $a^-(f):\mathcal F_{n+1}(H)\to\mathcal F_{n}(H)$ are bounded. Since $a^-(Jf)$ is the adjoint of $a^+(f)$, both operators $a^+(f)$ and $a^-(f)$ are closed. Now the statement follows from the closed graph theorem.
\end{proof}

By analogy with Subsection \ref{utkfwedt7}, we will now introduce, for each $f^{(2)}\in H^{\otimes 2}$,  operators
$a^{+-}(f^{(2)})$, $a^{++}(f^{(2)})$, and $a^{--}(f^{(2)})$ on $\mathcal F_{\mathrm{fin}} (H)$. For any $f_1,f_2\in H$ and $g^{(n)}\in\mathcal F_n(H)$ with $n\ge1$, we get, by \eqref{yre6ueu} and \eqref{yjqdse6ytqsd},
$$
a^+(f_1)a^-(f_2)g^{(n)}=\mathbb P_nl^{+-}(f_1\otimes f_2)\mathbb T_ng^{(n)}. $$
Hence, for each $f^{(2)}\in H^{\otimes 2}$, we define a linear operator $a^{+-}(f^{(2)})$ on $\mathcal F_{\mathrm{fin}} (H)$ by setting
$$a^{+-}(f^{(2)})\Omega:=0,\quad a^{+-}(f^{(2)})g^{(n)}:=\mathbb P_nl^{+-}(f^{(2)})\mathbb T_ng^{(n)},\quad g^{(n)}\in\mathcal F_n(H),\ n\in\mathbb N.$$
Note that, for a fixed $G\in \mathcal F_{\mathrm{fin}} (H)$, the mapping
\begin{equation}\label{xe5w8}
 H^{\otimes 2}\ni f^{(2)}\mapsto a^{+-}(f^{(2)})G\in\mathcal F (H)
 \end{equation}
is continuous.

Let a sequence $(A_n)_{n=1}^\infty$ and an operator $A$ be from $\mathcal L(\mathcal F_{\mathrm{fin}} (H))$. As usual, we will say that $A_n$ strongly converges to $A$ on $\mathcal F_{\mathrm{fin}} (H)$  if for each fixed $G\in \mathcal F_{\mathrm{fin}} (H)$, we have $\lim_{n\to\infty}A_nG= AG$ in the topology of $ \mathcal F_{\mathrm{fin}} (H)$.

Then the continuity of the mapping \eqref{xe5w8} implies
the decomposition
$$a^{+-}(f^{(2)}) =\sum_{i,j}\langle f^{(2)},e_i\otimes e_j\rangle a^+(e_i)a^-(e_j),$$
where the series strongly converges on
$\mathcal F_{\mathrm{fin}} (H)$.
 This also immediately implies
\begin{equation}\label{f65e}
a^{+-}(f^{(2)})^*=a^{+-}(\mathbb Sf^{(2)}).
\end{equation}

Similarly, for each $f^{(2)}\in H^{\otimes 2}$, we define a linear operator $a^{++}(f^{(2)})$ on $\mathcal F_{\mathrm{fin}} (H)$ by
\begin{equation}\label{r6w563i}
a^{++}(f^{(2)})\Omega:=\mathbb P_2f^{(2)},\quad a^{++}(f^{(2)})g^{(n)}=\mathbb P_{n+2}(l^{++}(f^{(2)}) g^{(n)}),\quad g^{(n)}\in\mathcal F_n(H),\ n\in\mathbb N.
\end{equation}

 Finally, to construct an operator $a^{--}(f^{(2)})$, we proceed as follows. For $f_1,f_2\in H$ and $g^{(n)}\in\mathcal F_n(N)$, $n\ge 2$, we get
\begin{align*}
&a^-(f_1)a^-(f_2)g^{(n)}=a^-(f_1)\mathbb P_{n-1}l^-(f_2)\mathbb T_ng^{(n)}\\
&\quad= a^-(f_1)(l^-(f_2)\restriction_H\otimes\mathbb P_{n-1})\mathbb T_ng^{(n)}\\
&\quad= \mathbb P_{n-2}l^-(f_1)\mathbb T_{n-1}(l^-(f_2)\restriction_H\otimes\mathbb P_{n-1})\mathbb T_ng^{(n)}\\
&\quad= \mathbb P_{n-2}l^-(f_1)(l^-(f_2)\restriction_H\otimes(\mathbb T_{n-1}\mathbb P_{n-1}))\mathbb T_ng^{(n)}\\
&\quad= \mathbb P_{n-2}l^-(f_1)l^-(f_2)(\mathbf 1_H\otimes(\mathbb T_{n-1}\mathbb P_{n-1}))\mathbb T_ng^{(n)}\\
&\quad= \mathbb P_{n-2}l^{--}(f_1\otimes f_2)(\mathbf 1_H\otimes(\mathbb T_{n-1}\mathbb P_{n-1}))\mathbb T_ng^{(n)}.
\end{align*}
Thus, for $f^{(2)}\in H^{\otimes 2}$, we define a linear operator $a^{--}(f^{(2)})$ by
\begin{align*}
&a^{--}(f^{(2)})\Omega:=0,\quad a^{--}(f^{(2)})g:=0,\quad g\in H,\\
&a^{--}(f^{(2)})g^{(n)}:=\mathbb P_{n-2}l^{--}(f^{(2)})(\mathbf 1_H\otimes(\mathbb T_{n-1}\mathbb P_{n-1}))\mathbb T_ng^{(n)},\quad g^{(n)}\in\mathcal F_n(H),\quad n\ge2.
\end{align*}
We easily see that  the above statements related to the operator $a^{+-}(f^{(2)})$ remain true (with obvious changes) for $a^{++}(f^{(2)})$ and $a^{--}(f^{(2)})$. In particular,
\begin{align*}
a^{++}(f^{(2)}) &=\sum_{i,j}\langle f^{(2)},e_i\otimes e_j\rangle a^+(e_i)a^+(e_j),\\
a^{--}(f^{(2)}) &=\sum_{i,j}\langle f^{(2)},e_i\otimes e_j\rangle a^-(e_i)a^-(e_j),
\end{align*}
where the series strongly converge on
$\mathcal F_{\mathrm{fin}} (H)$. Hence,
\begin{equation}\label{te5w53}
a^{++}(f^{(2)})^*=a^{--}(\mathbb Sf^{(2)}),\quad f^{(2)}\in H^{\otimes 2}.\end{equation}

By using formulas \eqref{f65e} and \eqref{te5w53}, analogously to the proof of
Proposition \ref{yedri751r}, we conclude the following proposition.

\begin{proposition}
For each $f^{(2)}\in H^{\otimes 2}$, we have
$$a^{+-}(f^{(2)}),\, a^{++}(f^{(2)}),\, a^{--}(f^{(2)})\in \mathcal L(\mathcal F_{\mathrm{fin}} (H)).$$
\end{proposition}

Assume that there exists an operator $\widetilde T\in\mathcal L(H^{\otimes 2})$ that satisfies the following identity:
\begin{equation}\label{re643}
\langle Tf_1\otimes f_2,f_3\otimes f_4\rangle=\langle \widetilde T f_3\otimes f_1, f_4\otimes f_2\rangle,\quad f_1,f_2,f_3,f_4\in H.\end{equation}
 Observe that identity \eqref{re643} does not necessarily identify a bounded linear operator $\widetilde T$, but in all known examples $\widetilde T$ indeed exists. Note also that if $\widetilde T$ exists, then it is obviously unique.

The following theorem states the commutation relations that the creation and annihilation operators satisfy   on the $T$-deformed Fock space.

\begin{theorem}\label{rdw6uged}
 For any $f,g\in H$,
\begin{equation}\label{e5uw732w7}
a^-(f)a^+(g)=a^{+-}(\widetilde T f\otimes g)+\langle f,g\rangle.
\end{equation}

Further let $f^{(2)}\in H^{\otimes 2}$. Then
\begin{equation}\label{fetdr56}
a^{++}(f^{(2)})=\mathbf 0\ \Leftrightarrow\  f^{(2)}\in \ker(\mathbf 1+T)
\end{equation}
 and
\begin{equation}\label{cftyqwf6}
a^{--}(f^{(2)})=\mathbf 0\ \Leftrightarrow\  \mathbb Sf^{(2)}\in \ker(\mathbf 1+T).
\end{equation}

Moreover, if $f^{(2)}\in \ker(\mathbf 1-T^2)$, then
\begin{equation}\label{ydqs6ri}
a^{++}(f^{(2)})=a^{++}(Tf^{(2)}), \end{equation}
and if $\mathbb Sf^{(2)}\in \ker(\mathbf 1-T^2)$, then
 \begin{equation}\label{rtdSC}
 a^{--}(f^{(2)})=a^{--}(\widehat Tf^{(2)}), \end{equation}
where
\begin{equation}\label{yqds6wsq}
\widehat T:=\mathbb ST\mathbb S. \end{equation}
\end{theorem}

\begin{proof} Formula \eqref{e5uw732w7} follows from \cite{BozSpe} (where it is written through an orthonormal basis in $H^{\otimes 2}$), see also \cite{Krolak} for the definition of the operator $\widetilde T$ in the basis-free form.

By \eqref{vydr6ei6} and \eqref{r6w563i}, we have
$$ a^{++}(f^{(2)})=a^{++}(\mathbb P_2f^{(2)}).$$
Hence, if $f^{(2)}\in \ker(\mathbf 1+T)=\ker(\mathbb P_2)$, then $a^{++}(f^{(2)})=\mathbf 0$, and if $f^{(2)}\not\in \ker(\mathbf 1+T)$, then \eqref{r6w563i} implies $a^{++}(f^{(2)})\Omega\ne0$.
Thus, \eqref{fetdr56} holds.
Formula \eqref{cftyqwf6} follows from \eqref{te5w53} and \eqref{fetdr56}.

 Formula \eqref{ydqs6ri}  follows from  \eqref{fetdr56} and  \eqref{KerrT}. Finally, by \eqref{KerrT} and \eqref{yqds6wsq},
$$\mathbb S \ker(\mathbf 1+T)=\big\{f^{(2)}-\widehat T f^{(2)}\mid \mathbb Sf^{(2)}\in \ker(\mathbf 1-T^2)\big\}.$$
Hence, formula \eqref{rtdSC} follows from \eqref{cftyqwf6}.
\end{proof}

\begin{remark}\label{d56w3ef} In view of \eqref{KerrT} and \eqref{fetdr56}, formulas
\eqref{ydqs6ri} and \eqref{rtdSC} describe all possible commutation relations between two creation operators or two annihilation operators, respectively.
\end{remark}

\begin{remark}
If the operator $T$ is  unitary, then $\ker(\mathbf 1-T^2)=H^{\otimes 2}$, hence equalities \eqref{ydqs6ri}, \eqref{rtdSC} hold for all $f^{(2)}\in H^{\otimes 2}$, in particular, for all $f^{(2)}\in\mathcal F_2(H)$.
\end{remark}

For  $A\in\mathcal L(H^{\otimes 2})$, we write
\begin{equation}\label{indexation}
A_{ij}^{kl}:=\langle Ae_i\otimes e_j,e_k\otimes e_l\rangle.
\end{equation}
Note that
\begin{equation}
 \label{indexation1}\widetilde T_{ij}^{kl}=T_{jl}^{ik},\quad
\widehat T_{ij}^{kl}=\overline{T_{ji}^{lk}}.
\end{equation}

The following corollary in immediate.

\begin{corollary}\label{uqrd7r}
We have
\begin{equation}\label{tesw5uwd}
a^-(e_i)a^+(e_j)=\sum_{k,l}\widetilde T_{ij}^{kl}a^+(e_k)a^-(e_l)+\delta_{ij},
\end{equation}
where $\delta_{ij}$ is the Kronecker delta. Furthermore, if $e_i\otimes e_j\in\ker(\mathbf 1-T^2)$, then
\begin{align}\label{rte64d}
a^+(e_i)a^+(e_j)&=\sum_{k,l} T_{ij}^{kl}a^+(e_k)a^+(e_l),\\
a^-(e_j)a^-(e_i)&=\sum_{k,l} \widehat T_{ji}^{kl}a^-(e_k)a^-(e_l).\label{d54w7ws}
\end{align}
In formulas \eqref{tesw5uwd}--\eqref{d54w7ws}, the series converge strongly on $\mathcal F_{\mathrm{fin}}(H)$.
\end{corollary}

\subsection{Wick algebras}

We finish this section with a short discussion of Wick algebras. Assume that
the operator $\widetilde T$ has the following property: for any $i,j$, only a finite number of $\widetilde T_{ij}^{kl}$ are not equal to zero. (We will say that the operator $\widetilde T$ has a finite matrix.)
Let $\mathbf A$ denote the complex algebra generated by the operators $a^+(e_i)$, $a^-(e_j)$ ($i,j\ge1$) and the identity operator.
 Then the commutation relation \eqref{tesw5uwd} implies that each element of this algebra can be represented as a linear combination of the identity operator and products of creation and annihilation operators in the Wick form:
$$a^+(e_{i_1})\dotsm a^+(e_{i_k})a^-(e_{j_1})\dotsm a^-(e_{j_l}), \quad k,l\ge0,\ k+l\ge1,$$
i.e., all creation operators are on the right and all annihilation operators are on the left. This is why one calls $\mathbf A$ a Wick algebra, see e.g.\ \cite{JoProSa,jorgensenschmittwernerJFA1995,Krolak, Krolak2}.

In the case where the  matrix of the operator $\widetilde T$ is not finite, one can proceed as follows. First, let us recall that if $\mathcal H_1$ and $\mathcal H_2$ are Hilbert spaces, then $\mathcal L(\mathcal H_1,\mathcal H_2)$, the space of all bounded linear operators from $\mathcal H_1$ into $\mathcal H_2$, is complete with respect to the strong convergence of bounded linear operators. Furthermore, an immediate consequence of the uniform boundedness principle states that, if $(A_n)_{n=1}^\infty$ and $(B_n)_{n=1}^\infty$ are  sequences in $\mathcal L(\mathcal H_1,\mathcal H_2)$ and $\lim_{n\to\infty}A_n=A$, $\lim_{n\to\infty}B_n=B$, then
$\lim_{n\to\infty}A_nB_n=AB$ (all limits being understood in the strong sense.)  These statements immediately imply the following lemma.

\begin{lemma}\label{dse5w6ue}
Let $\mathcal A\subset\mathcal L(\mathcal F_{\mathrm{fin}}(H))$. Let $\overline {\mathcal A}$ denote the closure of $\mathcal A$ with respect to the strong convergence on $\mathcal F_{\mathrm{fin}}(H)$.
Then $\overline {\mathcal A}\subset\mathcal L(\mathcal F_{\mathrm{fin}}(H))$.
Furthermore, if $\mathcal A$ is an algebra (with respect to  addition and product of operators), then $\overline{\mathcal A}$ is also an algebra.
\end{lemma}

Define the algebra $\mathbf A$ just as in the case where $\widetilde T$ had a finite matrix. Let $\mathbf W$ denote the subset of $\mathbf A$ that consists of all elements of $\mathbf A$ that can be represented as a linear combination of the identity operator and products of creation and annihilation operators in the Wick form. (Note that $\mathbf W$ is not anymore an algebra.) Let $\overline{\mathbf A}$ and $\overline{\mathbf W}$
denote the closures of $\mathbf A$ and $\mathbf W$ with respect to the strong convergence on $\mathcal F_{\mathrm{fin}}(H)$. Then, by Lemma~\ref{dse5w6ue}, $\overline {\mathbf A}\subset\mathcal L(\mathcal F_{\mathrm{fin}}(H))$ and $\mathbf A$ is an algebra. By formula \eqref{tesw5uwd}, we get $\overline{\mathbf A}=\overline{\mathbf W}$. Hence, in this case we may also think of $\mathbf A$ as a Wick algebra.

\section{Multicomponent quantum systems}\label{e4q4gdwyfd}

We will now discuss the  commutation relations for multicomponent quantum systems, in particular, for non-Abelian anyons.

Let $X:=\mathbb R^d$ ($d\in\mathbb N$) and let $V:=\mathbb C^m$, where $m\in\mathbb N$, $m\ge2$. We choose $H=L^2(X\to V)$, the $L^2$-space of $V$-valued functions on $X$. Here, as a reference measure, we chose the Lebesgue measure $dx$ on the Borel $\sigma$-algebra of $X$. Note that the space $H$ is the complexification of $L^2(X\to\mathbb R^m)$, the $L^2$-space of $\mathbb R^m$-valued functions on $X$. Note also that $H$ can  be naturally identified with the tensor product $L^2(X)\otimes V$, where $L^2(X)$ is the $L^2$-space of complex-valued functions on $X$.

Let
\begin{equation}\label{dr6e6ue}
X^{(2)}:=\{(x,y)\in X^2\mid x^1\ne y^1\},
\end{equation} where $x^i$ denotes the $i$th coordinate of $x$. Note that
$X^2\setminus X^{(2)}=\{(x,y)\in X^2\mid x^1= y^1\}$ is a set of zero $dx\,dy$-measure. Hence,
$$H^{\otimes 2}=L^2(X^{(2)}\to V^{\otimes 2}) =
L^2(X^{(2)})\otimes V^{\otimes 2}.$$
Similarly, for $n\ge3$, we have
$$H^{\otimes n}=L^2(X^{(n)}\to V^{\otimes n}) =
L^2(X^{(n)})\otimes V^{\otimes n},$$
where $X^{(n)}:=\{(x_1,\dots,x_n)\in X^n\mid x_i^1\ne x_j^1\text{ if }i\ne j\}$.

Below, for $(x,y)\in X^{(2)}$, we will write $x<y$ or $x>y$ if $x^1<y^1$ or $x^1>y^1$, respectively.

Consider $\mathcal L(V^{\otimes 2})$, the space of linear operators on $V^{\otimes 2}$, equivalently $m^2\times m^2$ matrices with complex entries. Consider a measurable mapping
$$X^{(2)}\ni(x,y)\mapsto C_{x,y}\in \mathcal L(V^{\otimes 2})$$
that satisfies the following assumptions: for each $(x,y)\in X^{(2)}$, $\|C_{x,y}\|\le1$ and $C_{x,y}^*=C_{y,x}$.  Define a linear operator $T$ on $H^{\otimes 2}$ by \eqref{T-C-map}.
As easily seen, the operator $T$ is bounded with $\|T\|\le1$ and self-adjoint.

\begin{lemma}\label{5e3ss}
The operator $T$ satisfies the Yang--Baxter equation \eqref{braid} if and only if the following equation  holds on $V^{\otimes 3}$ for a.a.\ $(x,y,z)\in X^{(3)}$:
\begin{equation}
C_{x,y}^{1,2}C_{x,z}^{2,3}C_{y,z}^{1,2}=C_{y,z}^{2,3}C_{x,z}^{1,2}C_{x,y}^{2,3}
\label{YBxyz}
\end{equation}
 Here $C_{v,w}^{k,k+1}$, $k=1,2$, denotes
the operator $C_{v,w}$ acting on the $k$th and $(k+1)$th components of the tensor product $ V^{\otimes 3}$.
\end{lemma}

\begin{remark}
Equation \eqref{YBxyz} is a spectral quantum Young--Baxter equation, see e.g.\ \cite[Section 6]{LM} and the references therein.
\end{remark}

\begin{proof}[Proof of Lemma \ref{5e3ss}] For the reader's convenience, we will prove this rather obvious lemma.
For $g\in L^2(X^{(3)})$ and $v\in V^{\otimes 3}$, we have
\begin{align*}
&(T\otimes\mathbf 1_H)(g\otimes v)(x,y,z)=g(y,x,z)C_{x,y}^{1,2}v,\\
&(\mathbf 1_H\otimes T)(T\otimes\mathbf 1_H)(g\otimes v)(x,y,z)=C_{y,z}^{2,3}(T\otimes\mathbf 1_H)(g\otimes v)(x,z,y)=g(z,x,y)C_{y,z}^{2,3}C^{1,2}_{x,z}v,\\
&(T\otimes \mathbf 1_H)(\mathbf 1\otimes T)(T\otimes \mathbf 1_H)(g\otimes v)(x,y,z)=C^{1,2}_{x,y}(\mathbf 1_H\otimes T)(T\otimes\mathbf 1_H)(g\otimes v)(y,x,z)\\& \quad= g(z,y,x)C_{x,y}^{1,2}C_{x,z}^{2,3}C_{y,z}^{1,2}v,
\end{align*}
and analogously
$$(\mathbf 1_H\otimes T)(T\otimes \mathbf 1_H)(\mathbf 1_H\otimes T)(g\otimes v)(x,y,z)=g(z,y,x)C_{x,y}^{1,2}C_{x,z}^{2,3}C_{y,z}^{1,2}v.\qedhere
 $$
\end{proof}

\begin{theorem}\label{vytr7i5}
Let $n\ge2$ and let $f^{(n)}\in H^{\otimes n}$. Then $f^{(n)}\in \mathcal F_n(H)$ if and only if, for each $i\in\{1,\dots,n-1\}$ and a.a.\ $(x_1,\dots,x_n)\in X^{(n)}$, we have
\begin{align}
&f^{(n)}(x_1,\dots,x_n)-C_{x_i,x_{i+1}}^{i,i+1}f^{(n)}(x_1,\dots,x_{i-1},x_{i+1},x_i,x_{i+2},\dots,x_n)\notag\\
&\quad \in V^{\otimes (i-1)}\otimes \operatorname{ran}(\mathbf 1_{V^{\otimes 2}}-C_{x_i,x_{i+1}}C_{x_i,x_{i+1}}^*)\otimes V^{\otimes(n-i-1)}.\label{rs5w5}\end{align}
Furthermore,  if condition \eqref{rs5w5} is satisfied for some $i\in\{1,\dots,n-1\}$ and $(x_1,\dots,x_n)\in X^{(n)}$, then it is automatically satisfied for this $i$ and $(x_1,\dots,x_{i-1},x_{i+1},x_i,x_{i+2},\dots,x_n)\in X^{(n)}$.
\end{theorem}

\begin{remark}
In view of the last statement of Theorem \ref{vytr7i5}, in oder to check whether  a given $f^{(n)}\in H^{\otimes n}$ belongs to $\mathcal F_n(H)$, it is sufficient to check  condition  \eqref{rs5w5} for all $i=1,\dots,n-1$ and a.a.\ $(x_1,\dots,x_n)\in X^{(n)}$ with $x_1<x_2<\dots<x_n$.

\end{remark}

In order to prove Theorem \ref{vytr7i5}, we will need the following two lemmas.

\begin{lemma}\label{ftyqfdyfdx}
Let $C\in\mathcal L(V^{\otimes 2})$ and let $w\in V^{\otimes 2}$.
 Then  $w\in \ker(\mathbf 1-CC^*)$ if and only if $C^*w\in \ker(\mathbf 1-C^*C)$. Moreover, the mappings
$$C^*:\ker(\mathbf 1-CC^*)\to \ker(\mathbf 1-C^*C),\quad C:\ker(\mathbf 1-C^*C)\to \ker(\mathbf 1-CC^*) $$
are bijective and inverse of each other.
\end{lemma}

\begin{proof}
Let $w\in \ker(\mathbf 1-CC^*)$, $w\ne0$. Then $w=CC^*w$, hence $C^*w\ne0$. Furthermore, $C^*w=C^*CC^*w$, hence $C^*w\in  \ker(\mathbf 1-C^*C)$. Therefore,
\begin{equation}\label{t54w335w}
C^*:\ker(\mathbf 1-CC^*)\to \ker(\mathbf 1-C^*C)\end{equation}
 is an injective mapping.  Swapping $C$ and $C^*$, we conclude that
\begin{equation}\label{xetrw5y4i}
C:\ker(\mathbf 1-C^*C)\to \ker(\mathbf 1-CC^*)\end{equation}
is an injective mapping. Finally, for $w\in \ker(\mathbf 1-CC^*)$, we have $w=CC^*w$ and for $v\in \ker(\mathbf 1-C^*C)$, we have $v=C^*Cv$. Hence, both mappings \eqref{t54w335w} and
\eqref{xetrw5y4i} are bijective and inverse of each other.
\end{proof}

\begin{lemma}\label{xtes67e54}
Let the conditions of Lemma {\rm \ref{ftyqfdyfdx}} be satisfied. Then, for any $u,v\in V^{\otimes 2}$, we have $u-Cv\in\operatorname{ran}(\mathbf 1-CC^*)$ if and only if $v-C^*u\in \operatorname{ran}(\mathbf 1-C^*C)$.
\end{lemma}

\begin{proof} Assume $u-Cv\in\operatorname{ran}(\mathbf 1-CC^*)$.
Then,
$$(u,w)-(Cv,w)=(u-Cv,w)=0\quad\text{for all }w\in\ker(\mathbf 1-CC^*).$$
Since $w=CC^*w$, we conclude:
$$0=(u,CC^*w)-(Cv,w)=(C^*u-v,C^*w)\quad\text{for all }w\in\ker(\mathbf 1-CC^*).$$
Hence, by Lemma \ref{ftyqfdyfdx},
 $$(C^*u-v,w)=0\quad\text{for all }w\in\ker(\mathbf 1-C^*C),$$
which implies  $C^*u-v\in \operatorname{ran}(\mathbf 1-C^*C)$. The converse implication is obvious.
\end{proof}

We can now proceed with the proof of Theorem \ref{vytr7i5}.

\begin{proof}[Proof of Theorem \ref{vytr7i5}]  In view of Theorem~\ref{se5sw5yw5wq}, it is sufficient to prove the result for $n=2$.
By the definition of $T$, we have, for each $f^{(2)}\in H^{\otimes 2}$,
$$\big((\mathbf 1-T^2)f^{(2)}\big)(x,y)=f^{(2)}(x,y)-C_{x,y}C_{y,x}f^{(2)}(x,y)=(\mathbf 1-C_{x,y}C_{x,y}^*)f^{(2)}(x,y).$$
From here we easily conclude that
$$\overline{\operatorname{ran}(\mathbf 1-T^2)}=\big\{f^{(2)}\in H^{\otimes 2}\mid
f(x,y)\in\operatorname{ran}(\mathbf 1-C_{x,y}C_{x,y}^*)\text{ for a.a.\ }(x,y)\in X^{(2)}\big\}.$$
Theorem~\ref{se5sw5yw5wq} now implies that, for each $f^{(2)}\in H^{\otimes 2}$, we have $f^{(2)}\in\mathcal F_2(H)$ if and only if
\begin{equation}\label{fyd5u6w6u4}
f(x,y)-C_{x,y}f(y,x)\in\operatorname{ran}(\mathbf 1-C_{x,y}C_{x,y}^*)\quad\text{for a.a.\ }(x,y)\in X^{(2)},\end{equation}
It follows from Lemma \ref{xtes67e54}  that  if condition \eqref{fyd5u6w6u4} is satisfied for some $(x,y)\in X^{(2)}$, then it is automatically satisfied for the point $(y,x)\in X^{(2)}$.\end{proof}

The following immediate corollary gives an explicit form of $\mathfrak P$, the orthogonal projection of $H^{\otimes 2}$ onto $\mathcal F_2(H)$ (compare with Proposition~\ref{rw5uwu83xx}).

\begin{corollary}\label{rewa5w}
For $(x,y)\in X^{(2)}$, $x<y$, denote by $P_{x,y}$ the orthogonal projection of the space $V^{\otimes 2}\oplus V^{\otimes 2}$ onto the subspace
$$\big\{(u,v)\in V^{\otimes 2}\oplus V^{\otimes 2}\mid u-C_{x,y}v\in\operatorname{ran}(\mathbf 1-C_{x,y}C_{x,y}^*)\big\}.$$
Further for
$P_{x,y}(u,v)=(w_1,w_2)$, with $w_1,w_2\in V^{\otimes 2}$, we denote
$P^1_{x,y}(u,v):=w_1$ and $P^2_{x,y}(u,v):=w_2$, i.e., $P^1_{x,y}(u,v)$ and $P^2_{x,y}(u,v)$ are the first and second $V^{\otimes 2}$-coordinates of the vector $P_{x,y}(u,v)$. Then $\mathfrak P$, the orthogonal projection of $H^{\otimes 2}$ onto $\mathcal F_2(H)$, has the following form: for $(x,y)\in X^{(2)}$ with $x<y$,
$$(\mathfrak Pf^{(2)})(x,y)=P_{x,y}^1\big(f^{(2)}(x,y),f^{(2)}(y,x)\big),\quad (\mathfrak Pf^{(2)})(y,x)=P_{x,y}^2\big(f^{(2)}(x,y),f^{(2)}(y,x)\big).$$
\end{corollary}

Let us now describe $\ker(\mathbf 1+T)=\ker(\mathfrak P)$.

\begin{proposition}\label{t5we64u4}
We have
\begin{align}
&\ker(\mathbf 1+T)=\big\{ f^{(2)}-Tf^{(2)}\mid f^{(2)}\in H^{\otimes 2} \text{\ \rm   and }  f^{(2)}(x,y)\in \ker(\mathbf 1-C_{x,y}C_{x,y}^*)\notag\\
&\qquad\quad\text{\rm for a.a. }(x,y)\in X^{(2)}\big\}\label{vyre6w64u}\\
%&\quad=\big\{f^{(2)}\in H^{\otimes 2}\mid \text{\rm  $f^{(2)}(x,y)\in\ker(\mathbf 1-C_{x,y}C_{x,y}^*)$ }\notag\\
%&\qquad  \quad\text{\rm and $f^{(2)}(x,y)=-C_{x,y}f^{(2)}(y,x)$},\text{ \rm for a.a.\ $(x,y)\in X^{(2)}$}\big\}\label{intermed}\\
&\quad=\big\{f^{(2)}\in H^{\otimes 2}\mid \text{\rm  $f^{(2)}(x,y)\in\ker(\mathbf 1-C_{x,y}C_{x,y}^*)$ if $x<y$}\notag\\
&\qquad  \quad\text{\rm and $f^{(2)}(x,y)=-C_{x,y}f^{(2)}(y,x)$ if $x>y$} \text{ \rm for a.a.\ $(x,y)\in X^{(2)}$}\big\}.\label{diygudd}
\end{align}
Formula \eqref{diygudd} remains true if we swap the conditions $x<y$ and $x>y$.
\end{proposition}

\begin{remark}\label{yqtfdt7edqr}
Formula \eqref{diygudd} can be interpreted as follows: $\ker(\mathbf 1+T)$ consists of all functions of the form $f^{(2)}-Tf^{(2)}$, where $f^{(2)}\in H^{\otimes 2}$ satisfies the following assumption: for a.a.\ $(x,y)\in X^{(2)}$, $f^{(2)}(x,y)\in\ker(\mathbf 1-C_{x,y}C_{x,y}^*)$ if $x<y$ and $f^{(2)}(x,y)=0$ if $x>y$.
\end{remark}

\begin{proof}[Proof of Proposition \ref{t5we64u4}] Formula \eqref{vyre6w64u} follows immediately from \eqref{KerrT} and the equality $(\mathbf 1-T^2)f^{(2)}(x,y)=(\mathbf 1-C_{x,y}C_{x,y}^*)f^{(2)}(x,y)$.

Due to the inclusion $\ker(\mathbf 1+T)\subset\ker(\mathbf 1-T^2)$, we get
$$\ker(\mathbf 1+T)=\ker(\mathbf 1+T)\cap \ker(\mathbf 1-T^2),$$
or equivalently
\begin{align}
&\ker(\mathbf 1+T)=\big\{f^{(2)}\in H^{\otimes 2}\mid \text{\rm  $f^{(2)}(x,y)\in\ker(\mathbf 1-C_{x,y}C_{x,y}^*)$ }\notag\\
&\qquad  \quad\text{\rm and $f^{(2)}(x,y)=-C_{x,y}f^{(2)}(y,x)$},\text{ \rm for a.a.\ $(x,y)\in X^{(2)}$}\big\}.\label{intermed}
\end{align}
By Lemma~\ref{ftyqfdyfdx},  if  the relation
$$
f^{(2)}(x,y)=C_{x,y}C_{x,y}^*f^{(2)}(x,y),\quad f^{(2)}(y,x)=-C_{x,y}^*f^{(2)}(x,y).
$$
holds for $x<y$,  then it also holds for $x>y$. Hence,  formula \eqref{diygudd} follows from \eqref{intermed}.
\end{proof}

According to the general considerations in  Subsection \ref{uerqdqdf}, we can now construct  creation and annihilation operators in the $T$-deformed Fock space $\mathcal F(H)$. It should be noticed that the operator $\mathbb T_n$ given by formula \eqref{yfdqfqdr} (and used for the annihilation operators in formula \eqref{yjqdse6ytqsd}) has now the following form:
\begin{align*}
&(\mathbb T_nf^{(n)})(x_1,\dots,x_n)=f^{(n)}(x_1,\dots,x_n)\\
&\quad + \sum_{k=2}^n C_{x_1,x_2}^{1,2}C_{x_1,x_3}^{2,3}\dotsm C_{x_1,x_k}^{k-1,k}f^{(n)}(x_2,x_3,\dots,x_k,x_1,x_{k+1},\dots,x_n),\quad f^{(n)}\in H^{\otimes n}.
\end{align*}

Recall that, in Subsection \ref{uerqdqdf}, for the given operator $T\in\mathcal L(H^{\otimes 2})$, we defined the operator $\widetilde T$ through equality \eqref{re643} and the operator $\widehat T$ by \eqref{yqds6wsq}.   Similarly, for a linear operator $C\in\mathcal L(V^{\otimes 2})$, we define linear operators $\widetilde C, \widehat C\in \mathcal L(V^{\otimes 2})$.  (Note that, in the finite-dimensional setting, the operator $\widetilde C$ always exists.)

\begin{lemma}\label{yre6i64}
For $f^{(2)}\in H^{\otimes 2}$, we have
\begin{align}
(\widetilde T f^{(2)})(x,y)&=\widetilde C_{y,x}f^{(2)}(y,x),\label{tdfw}\\
 (\widehat T f^{(2)})(x,y)&=\widehat C_{y,x}f^{(2)}(y,x).\label{es4a}\end{align}
\end{lemma}

\begin{proof}
For $i=1,2,3,4$, let $f_i(x)=\varphi_i(x)u_i$, where $\varphi_i\in L^2(X)$ and $u_i\in V$.
Then
\begin{align*}
 \langle Tf_1\otimes f_2,f_3\otimes f_4\rangle&=
 \int_{X}f_1(y)f_2(x)f_3(x)f_4(y)\langle
 C_{x,y}u_1\otimes u_2,u_3,\otimes u_4\rangle\,dx\,dy\\
 &=
 \int_{X}f_1(y)f_2(x)f_3(x)f_4(y)\langle
 \widetilde C_{x,y}u_3\otimes u_1,u_4,\otimes u_2\rangle\,dx\,dy\\
 &=
 \int_{X}f_1(x)f_2(y)f_3(y)f_4(x)\langle
 \widetilde C_{y,x}u_3\otimes u_1,u_4,\otimes u_2\rangle\,dx\,dy\\
 &=
 \int_{X} (f_3\otimes f_1)(y,x)(f_4\otimes f_2)(x,y)
 \langle
 \widetilde C_{y,x}u_3\otimes u_1,u_4,\otimes u_2\rangle\,dx\,dy,
\end{align*}
which proves \eqref{tdfw}.

To prove \eqref{es4a}, we proceed as follows:
\begin{align*}
&(\mathbb S f_1\otimes f_2)(x,y)=\overline{\varphi_1(y)\varphi_2(x)}\, J(u_2\otimes u_1),\\
&(T\mathbb S f_1\otimes f_2)(x,y)= \overline{\varphi_1(x)\varphi_2(y)}\, C_{x,y} J(u_2\otimes u_1),\\
&(\mathbb S T\mathbb S f_1\otimes f_2)(x,y)= \varphi_1(y)\varphi_2(x) \widehat C_{y,x} u_1\otimes u_2,
\end{align*}
which implies \eqref{es4a}.
\end{proof}

To specialize the result of Theorem \ref{rdw6uged} to our current setting, it will be convenient for us to introduce formal operators of creation and annihilation at point $x\in X$. Let $f\in H=L^2(X\to V)$. Then
\begin{equation}\label{yte567}
f(x)=(\varphi_1(x),\varphi_2(x),\dots,\varphi_m(x)),\end{equation} where $\varphi_1,\varphi_2,\dots,\varphi_m\in L^2(X)$. For $i=1,2,\dots,m$ and $\varphi\in L^2(X)$, we denote
$$a^+_i(\varphi):=a^+(0,\dots,0,\underbrace{\varphi}_{\text{$i$th place}},0,\dots,0).$$
Then, for $f\in L^2(X\to V)$ of the form \eqref{yte567}, we get
\begin{equation}\label{dutre56e}
a^+(f)=\sum_{i=1}^m a^+_i(\varphi_i). \end{equation}
For  $i=1,2,\dots,m$ and $x\in X$, we formally define a creation operator $a^+_i(x)$ that satisfies
\begin{equation}\label{e4563uw}
a^+_i(\varphi)=\int_X \varphi(x) a^+_i(x)\,dx\quad \text{for all }\varphi\in L^2(X).\end{equation}
Thus, $a^+_i(x)$ can be formally interpreted as an operator-valued distribution. Next, we define a vector $a^+(x)$ of operator-valued distributions by
$$a^+(x):=(a^+_1(x),a_2^+(x),\dots,a^+_m(x)).$$
In other words, $a^+(x)$ has $n$ components, each of which is an operator-valued distribution.

We will formally operate with $a^+(x)$ as a usual vector from $V$. So, for a vector $v=(v_1,\dots,v_m)\in V$, the $\langle\cdot,\cdot\rangle$ product of $v$ and $a^+(x)$ is given by
$$\langle v,a^+(x)\rangle=\sum_{i=1}^m v_ia^+_i(x).$$
Hence,  for a fixed $x\in X$ and  a function $f(x)$ of the form \eqref{yte567}, we have
\begin{equation}\label{vcrt5w5uwu6}
\langle f(x),a(x)\rangle=\sum_{i=1}^m \varphi_i(x)a_i^+(x).\end{equation}
In view of formulas \eqref{dutre56e}--\eqref{vcrt5w5uwu6}, we get
$$a^+(f)=\int_X\langle f(x),a^+(x)\rangle\,dx.$$
We similarly define $a^-(x)$ satisfying
$$a^-(f)=\int_X\langle f(x),a^-(x)\rangle\,dx.$$

Next, for $x,y\in X$, we may formally use the tensor product of the `vectors' $a^+(x)$ and $a^-(y)$:
$$a^+(x)\otimes a^-(y)=\big(a_i^+(x) a_j^-(y)\big)_{i,j=1,\dots,m}.$$
Hence, for $f\in L^2(X\to V)$ of the form \eqref{yte567} and $g\in L^2(X\to V)$ of the form
$$g(y)=(\psi_1(y),\psi_2(y),\dots,\psi_m(y)),$$
 we get
\begin{align}
&\int_{X^2} \big\langle (f\otimes g)(x,y),a^+(x)\otimes a^-(y)\big\rangle\,dx\,dy\notag\\
&\quad=\sum_{i,j=1}^m\int_{X^2}\varphi_i(x)\psi_j(y) a_i^+(x) a_j^-(y) \,dx\,dy\notag\\
&\quad=\sum_{i,j=1}^m \int_X\varphi_i(x)a_i^+(x)\,dx\,\int_X\psi_j(y)a^-_j(y)\,dy\notag\\
&\quad=\sum_{i,j=1}^m a_i^+(\varphi_i)a_j^-(\psi_j)=a^+(f)a^-(g).\notag
\end{align}
Hence, for an arbitrary $f^{(2)}\in H^{\otimes 2}$, we will write
$$a^{+-}(f^{(2)})=\int_{X^2}\big\langle f^{(2)}(x,y), a^+(x)\otimes a^-(y)\big\rangle\,dx\,dy.$$
We will use similar notations for $a^{++}(f^{(2)})$,  $a^{--}(f^{(2)})$, and for a product $a^-(f)a^+(g)$ with $f,g\in H$.

Let
$$X^{(2)}\ni (x,y)\mapsto M_{x,y}\in\mathcal L(V^{\otimes 2})$$ be a measurable mapping with $\|M_{x,y}\|\le1$. Then, we will write, for $f^{(2)}\in H^{\otimes 2}$,
$$  \int_{X^2}\big\langle M_{x,y}f^{(2)}(x,y), a^+(x)\otimes a^-(y)\big\rangle\,dx\,dy = \int_{X^2}\big\langle f^{(2)}(x,y), M_{x,y}^Ta^+(x)\otimes a^-(y)\big\rangle\,dx\,dy, $$
where $A^T$ denotes the transpose of a matrix  $A$.

\begin{theorem} \label{yjsqrd6uq}
For any $f,g\in H$, we have
\begin{align}
&\int_{X^2}\big\langle (f\otimes g)(x,y), a^-(x)\otimes a^+(y)\big\rangle \,dx\,dy \notag\\
&\quad=\int_{X^2} \big\langle (f\otimes g)(x,y),\widetilde C^{\,T}_{x,y}a^+(y)\otimes a^-(x)\big\rangle\,dx\,dy+\int_X \langle f(x),g(x)\rangle\,dx.\label{s5342}
\end{align}

Further assume that $\ker(\mathbf 1+T)\ne\{0\}$ and let $f^{(2)}\in H^{\otimes 2}$.  If for a.a.\ $(x,y)\in X^{(2)}$, $f^{(2)}(x,y)\in \ker(\mathbf 1-C_{x,y}C_{x,y}^*)$, then
\begin{equation}\label{yqdrwd}
\int_{X^2}\big\langle f^{(2)}(x,y),a^+(x)\otimes a^+(y)\big\rangle
\,dx\,dy= \int_{X^2}\big\langle f^{(2)}(x,y), C^{T}_{y,x} a^+(y)\otimes a^+(x)\big\rangle \,dx\,dy,\end{equation}
and if for  a.a.\ $(x,y)\in X^{(2)}$, $(\mathbb Sf^{(2)})(x,y)\in \ker(\mathbf 1-C_{x,y}C_{x,y}^*)$, then
\begin{equation}\label{eri7se}
\int_{X^2}\big\langle f^{(2)}(x,y),a^-(x)\otimes a^-(y)\big\rangle
\,dx\,dy= \int_{X^2}\big\langle f^{(2)}(x,y), \widehat C^{\,T}_{x,y} a^-(y)\otimes a^-(x)\big\rangle \,dx\,dy.\end{equation}
\end{theorem}

\begin{proof}
The statement follows immediately from Theorem \ref{rdw6uged}, formula \eqref{vyre6w64u} from Proposition~\ref{t5we64u4} and Lemma~\ref{yre6i64}. \end{proof}

 \begin{remark}
 Let $A$ be a measurable subset of $X^2$ and assume that a function $f^{(2)}\in H^{\otimes 2}$ vanishes outside the set $A$. Then it is natural to write
$$a^{++}(f^{(2)})=\int_A f^{(2)}(x,y)a^+(x)a^+(y),\quad  a^{--}(f^{(2)})=\int_A f^{(2)}(x,y)a^+(x)a^+(y).$$
In view of \eqref{diygudd}  (see also Remark~\ref{yqtfdt7edqr}), formulas \eqref{yqdrwd} and \eqref{eri7se} can be {\it equivalently\/} written as follows.
Let $f^{(2)}\in H^{\otimes 2}$ be such that $f^{(2)}(x,y)=0$ for all $(x,y)\in X^{(2)}$ with $x>y$. If for a.a.\ $(x,y)\in X^{(2)}$ with $x<y$, we have
 $f^{(2)}(x,y)\in\ker(\mathbf 1-C_{x,y}C_{x,y}^*)$, then
\begin{equation}\label{h8t86}
\int_{\{x<y\}}\big\langle f^{(2)}(x,y),a^+(x)\otimes a^+(y)\big\rangle
\,dx\,dy= \int_{\{x<y\}}\big\langle f^{(2)}(x,y), C^{T}_{y,x} a^+(y)\otimes a^+(x)\big\rangle \,dx\,dy,
\end{equation}
and if for a.a\ $(x,y)\in X^{(2)}$ with $x>y$, we have $\mathbb S f^{(2)}(x,y)\in\ker(\mathbf 1-C_{x,y}C_{x,y}^*)$, then
\begin{equation}\label{tew6u54we}
\int_{\{x<y\}}\big\langle f^{(2)}(x,y),a^-(x)\otimes a^-(y)\big\rangle
\,dx\,dy= \int_{\{x<y\}}\big\langle f^{(2)}(x,y), \widehat C^{\,T}_{x,y} a^-(y)\otimes a^-(x)\big\rangle \,dx\,dy.
\end{equation}
If we swap  the conditions $x<y$ and $x>y$, the above results will remain true.
 \end{remark}

Theorem \ref{yjsqrd6uq} (and formulas \eqref{h8t86}, and \eqref{tew6u54we}) can be easily understood by using formal commutation relations between the creation and annihilation operators at point.

\begin{corollary}\label{formalCR}
The following formal commutation relations hold.

{\rm (i)}
 For all $(x,y)\in X^{(2)}$, we have
$$ a^-(x)\otimes a^+(y)=\widetilde C^{\,T}_{x,y}a^+(y)\otimes a^-(x)+\delta(x-y)\Delta.$$
Here $\Delta:=(\delta_{ij})_{i,j=1,\dots,m}$ with $\delta_{ij}$ being the Kronecker delta, so that for any $f,g\in H$,
$$\int_{X^2} \big\langle (f\otimes g)(x,y), \delta(x-y)\Delta\big\rangle
\,dx\,dy=\int_X \langle f(x),g(x)\rangle\,dx.$$

{\rm (ii)}
For each $(x,y)\in X^{(2)}$ and a vector $v\in\ker(\mathbf 1-C_{x,y}C_{x,y}^*)$, we have
$$\langle v,a^+(x)\otimes a^+(y)\rangle=\langle C_{y,x}v, a^+(y)\otimes a^+(x)\big\rangle=\langle v,C^{T}_{y,x} a^+(y)\otimes a^+(x)\big\rangle,  $$
and for each $(x,y)\in X^{(2)}$ and a vector $v\in V^{\otimes 2}$ such that $\mathbb Sv\in \ker(\mathbf 1-C_{x,y}C_{x,y}^*)$, we have
$$\langle v,a^-(x)\otimes a^-(y)\rangle=\langle \widehat C_{x,y}v, a^-(y)\otimes a^-(x)\big\rangle=\langle v,\widehat C^{\,T}_{x,y} a^-(y)\otimes a^-(x)\big\rangle.  $$
Here $\mathbb S$ acts on the space $V^{\otimes 2}$.
\end{corollary}

In the case where the operator $T$ is unitary, formulas \eqref{yqdrwd}, \eqref{eri7se} hold for all $f^{(2)}\in H^{\otimes 2}$. Thus,  the commutation relations take the following  form.

\begin{corollary}\label{formalCRU}
Let $T$ be unitary. Then,  for all $(x,y)\in X^{(2)}$, we formally have:
\begin{align*}
&a^-(x)\otimes a^+(y)=\widetilde C^{\,T}_{x,y}a^+(y)\otimes a^-(x)+\delta(x-y)\Delta,\\
&a^+(x)\otimes a^+(y)= C^{T}_{y,x} a^+(y)\otimes a^+(x),\\
&a^-(x)\otimes a^-(y)=\widehat C^{\,T}_{x,y} a^-(y)\otimes a^-(x).
\end{align*}
\end{corollary}

\section{Examples}
\label{sec-examp}

In this section, we will consider several particular examples of the operator $T$ associated with a multicomponent quantum system  and explicitly compute the corresponding Fock space and commutation relations between creation and annihilation operators. In all but the very last example, the operator $T$ will be constructed through a single linear operator $C$ on $V^{\otimes2}$  which satisfies the (constant) Yang--Baxter equation on $V^{\otimes3}$:
\begin{equation}  \label{YBnumer}
C^{1,2} C^{2,3} C^{1,2} = C^{2,3} C^{1,2} C^{2,3}.
\end{equation}
  We restrict ourselves to  $V = \mathbb{C}^2$, in which case all solutions of the (equivalent form of the) Yang--Baxter equation are classified in \cite{hiet}, see also the earlier PhD thesis \cite{Lyubashenko}.

We will denote by $(e_1,e_2)$ the standard orthonormal basis of $V=\mathbb{C}^2$, and by $(e_{11},e_{12},e_{21},e_{22})$, with $e_{ij}:=e_i\otimes e_j$, the corresponding orthonormal basis of $V^{\otimes 2}$.  In this basis, we will identify linear operators on $V\otimes V$ with matrices acting on column vectors.  By \eqref{indexation1}, if
\begin{equation}\label{ctrs6uw}
\text{if }C=\left(\begin{matrix}
c_{11}^{11} & c_{11}^{12} & c_{11}^{21} & c_{11}^{22}\\
c_{12}^{11} & c_{12}^{12} & c_{12}^{21} & c_{12}^{22}\\
c_{21}^{11} & c_{21}^{12} & c_{21}^{21} & c_{21}^{22}\\
c_{22}^{11} & c_{22}^{12} & c_{22}^{21} & c_{22}^{22}
\end{matrix}\right),\text{ then }\widetilde C=\left(\begin{matrix}
c_{11}^{11} & c_{12}^{11}& c_{11}^{12}& c_{12}^{12} \\
 c_{21}^{11}& c^{11}_{22}  &c_{21}^{12} & c^{12}_{22}\\
 c_{11}^{21}&c_{12}^{21} &c_{11}^{22}  & c_{12}^{22} \\
  c_{21}^{21}&c_{22}^{21}  & c_{21}^{22}   & c_{22}^{22}
\end{matrix}\right). \end{equation}

For a function $f^{(n)}\in H^{\otimes n}$, we will denote by $f^{(n)}_{i_1i_2\dots i_n}\in L^2(X^{(n)})$ the $e_{i_1}\otimes e_{i_2}\otimes\dots\otimes e_{i_n}$ coordinate of $f^{(n)}$, where  $i_1,i_2,\dots,i_n\in\{1,2\}$.

%In what follows, we will denote by $c.l.s.$  the {\it closed linear span} (of a collection of elements of a vector space). Given $f\in  L_2 (X^2)$ we will use notation $f^t$ for the function defined by $f^t(x,y)=f(y,x)$.

\subsection{Spatially constant $C$}\label{C-const}

We start with the case where $C_{xy}$ is independent of spatial variables $%
x,y$, i.e., $C_{xy}=C$ for a fixed matrix $C= C^\ast$, $\|C\|\leq 1$. Then the operator $T$ satisfies  equation \eqref{braid} if and only if
the matrix $C$ satisfies requation \eqref{YBnumer}.

%Let us stress out also, that due to the fact that $C$ is independent on $(x,y)\in X^2$, in this section we will put relations between creation and %annihilation operators in ``global form'', not at point.

\begin{example}\label{ra4tqq}
Let us consider the  operator $C$  given by the matrix
\begin{equation*}
C=\left(
\begin{matrix}
k & 0 & 0 & 0 \\
0 & 0 & q & 0 \\
0 & \bar{q} & 0 & 0 \\
0 & 0 & 0 & k%
\end{matrix}%
\right).
\end{equation*}%
Here $k\in (-1,1)$ and $q\in\mathbb{C}$, $|q|=1$. Then \begin{equation*}
\mathbf 1-C^{2}=\left(
\begin{matrix}
1-k^{2} & 0 & 0 & 0 \\
0 & 0 & 0 & 0 \\
0 & 0 & 0 & 0 \\
0 & 0 & 0 & 1-k^{2}%
\end{matrix}%
\right),
\end{equation*}%
which implies
\begin{align}
\ran(\mathbf 1-C^2)&=\ls\{e_{11},e_{22}\},\label{xesa5yw3}\\
\ker(\mathbf 1-C^2)&=\ls\{e_{12},e_{21}\}.\label{we4qy}
\end{align}
Here $\ls$ denotes the linear span.
For $f^{(2)}\in H^{\otimes2}$,
\begin{align*}
f^{(2)}(x,y)-Cf^{(2)}(y,x)&=\big(f_{11}^{(2)}(x,y)-kf_{11}^{(2)}(y,x)\big)e_{11}+\big(f_{12}^{(2)}(x,y)-\bar q f_{21}^{(2)}(y,x)\big)e_{12}\\
&\quad+\big(f_{21}^{(2)}(x,y)-q f_{12}^{(2)}(y,x)\big)e_{21}+\big(f_{22}^{(2)}(x,y)-kf_{221}^{(2)}(y,x)\big)e_{22}.
\end{align*}
Hence, by \eqref{xesa5yw3}, the condition $f^{(2)}(x,y)-Cf^{(2)}(y,x)\in \ran(\mathbf 1-C^2)$ is equivalent to
\begin{equation}\label{rts54ywq5}
f_{21}^{(2)}(x,y)=q f^{(2)}_{12}(y,x).\end{equation}

By  Theorem~\ref{vytr7i5}, we now get the following explicit description of $\mathcal F_n(H)$.
Define
\begin{equation}\label{se53w272}
Q(1,2):=\bar q,\quad Q(2,1):=q.
\end{equation}
 Then for $n\ge 2$, $\mathcal F_n(H)$   consists of all functions $f^{(n)}\in H^{\otimes n}$ that satisfy a.e.\ the following symmetry condition:
\begin{equation}\label{s4q53y2}
f^{(n)}_{i_1\dots i_n}(x_1,\dots,x_n)=Q(i_k,i_{k+1})f^{(n)}_{i_1\dots i_{k-1}i_{k+1}i_ki_{k+2}\dots i_n}(x_1,\dots,x_{k-1},x_{k+1},x_{k},x_{k+2},\dots,x_n) \end{equation}
 for $k\in\{1,\dots,n-1\}$ and $i_1,\dots,i_n\in\{1,2\}$ with $i_k\ne i_{k+1}$.
In particular, the function $f^{(n)}\in\mathcal F_n(H)$ is completely identified
by its coordinates $f^{(n)}_{i_1\dots i_n}$ with $i_1\le i_2\le \dots\le i_n$.

By using Corollary \ref{rewa5w} and \eqref{xesa5yw3}, one can easily calculate $\mathfrak P$, the orthogonal projection  of $H^{\otimes 2}$ onto $\mathcal F_2(H)$:
\begin{align}
(\mathfrak P f^{(2)})(x,y)&=f_{11}^{(2)}(x,y)e_{11}+f_{22}^{(2)}(x,y)e_{22}+
\frac12\big(f_{12}^{(2)}(x,y)+\bar q f_{21}^{(2)}(y,x)\big)e_{12}\notag\\
&\quad+\frac12\big(f_{21}^{(2)}(x,y)+ q f_{12}^{(2)}(y,x)\big)e_{21}.\label{rsw54aq4}
\end{align}

To obtain the commutation relations between creation and annihilation
operators, we use Theorem \ref{yjsqrd6uq} and \eqref{we4qy}. Additionally to \eqref{se53w272}, set also
$$
Q(1,1)=Q(2,2):=k.$$
By \eqref{ctrs6uw}, we get $\widetilde C=C^*$. Hence, for all $\varphi,\psi\in L^2(X)$,
\begin{align}
&a^-_i(\varphi)a_i^+(\psi)=Q(i,j)a_i^+(\psi)a^-_i(\varphi)+\delta_{ij}\langle \varphi,\psi\rangle,\quad i,j\in\{1,2\},\notag\\
&a^+_i(\varphi)a_i^+(\psi)=Q(j,i)a_i^+(\psi)a^+_i(\varphi),\quad i\ne j,\notag\\
&a^-_i(\varphi)a_i^-(\psi)=Q(j,i)a_i^-(\psi)a^-_i(\varphi),\quad i\ne j.
\label{qds5qsd}
\end{align}

It can be shown that in this case there exists the universal $C^*$-algebra $\mathbf A$ generated by $a_i^{+}(\varphi)$, $a_i^{-}(\varphi)$, $i=1,2$, $\varphi\in L^2(X)$. Let also $\mathbf A_1$, $\mathbf A_2$ be the $C^*$-subalgebras of
$\mathbf A$ generated by $a_1^{+}(\varphi)$, $a_1^{-}(\varphi)$ and
$a_2^{+}(\varphi)$, $a_2^{-}(\varphi)$, respectively. Note that each $\mathbf{A}_i$ (i=1,2) is the $C^*$-algebras generated by the $k$-deformed commutation relations with $k\in (-1,1)$, see \cite{BS}.

One can construct the tensor product $\mathbf A_1\otimes\mathbf A_2$ and consider its Rieffel deformation, denoted by $\mathbf A_1\otimes_q \mathbf A_2$, see \cite{rieffel}. Then it turns out that the Fock representation of $\mathbf A$ can be realized as the composition of the canonical surjection $\Phi\colon \mathbf A\rightarrow \mathbf A_1\otimes_q \mathbf A_2$ and the Fock representation of $\mathbf A_1\otimes_q \mathbf A_2$.
This approach will give us a deeper insight into the structure of the Fock representation of $\mathbf{A}$.
\iffalse
We will now give a somewhat different description of the $*$-algebra $\mathbf{A}$ generated by the operators $a_i^{+}(\varphi)$ and $a_i^{-}(\varphi)$, $i=1,2$, $\varphi\in L^2(X)$, and the identity operator.  Namely, we will show that $\mathbf{A}$ is the Rieffel deformation of the tensor product of the $C^*$-algebras $\mathbf{A}_1$ and $\mathbf{A}_2$, see \cite{rieffel}. Here,  for $i=1,2$,
$\mathbf A_i$ is generated by   $a_i^{+}(\varphi)$ and $a_i^{-}(\varphi)$, $\varphi\in L^2(X)$, and the identity operator.
\fi

Below we will use the fact that any irreducible representation of $\mathbf{A}$ that possesses a vacuum vector annihilated by $a_i^{-}(\varphi)$, $i=1,2$, $\varphi\in L^2(X)$, is unitarily equivalent to the Fock representation, see \cite{jorgensenschmittwernerJFA1995}.

Let $K:=L^2(X)$. Construct the
 Fock space
$\mathcal{F}(K)=\bigoplus_{n=0}^\infty \mathcal F_n(K)$
corresponding to the Fock representation of the $k$-deformed commutation relations, and denote by $\Psi$ the vacuum vector in $\mathcal{F}(K)$, see \cite{BS} for details.
Let $a^{+}(\varphi)$, $a^{-}(\varphi)$ ($\varphi\in K$) be the corresponding  creation and annihilation operators on $\mathcal{F}(K)$.
We construct a unitary operator $U\colon\mathcal{F}(K)\rightarrow\mathcal{F}(K)$ by
\[
U\Psi=\Psi,\quad U \varphi^{(n)}=q^n \varphi^{(n)},\quad \varphi^{(n)}\in\mathcal{F}_n(K),\ n\in\mathbb{N}.
\]
Obviously,
\[
U a^{+}(\varphi)= q a^{+}(\varphi) U,\quad U a^{-}(\varphi)= \bar{q} a^{-}(\varphi) U.
\]
Define the space $\mathcal{F}:=\mathcal{F}(K)\otimes\mathcal{F}(K)$ and bounded linear operators operators $a_i^+(\varphi)$, $a_i^-(\varphi)$ ($\varphi\in K$) on $\mathcal F$ by
\begin{align}
a_1^{+}(\varphi)&:=a^{+}(\varphi)\otimes\mathbf{1}_{\mathcal{F}(K)},\quad a_2^{+}(\varphi):=U\otimes a^{+}(\varphi),\nonumber\\
a_1^{-}(\varphi)&:=a^{-}(\varphi)\otimes\mathbf{1}_{\mathcal{F}(K)},\quad a_2^{-}(\varphi):=U^*\otimes a^{-}(\varphi).\label{k_qFock}
\end{align}
It is easy to verify that these operators satisfy the commutation relation \eqref{qds5qsd}.  The  family $\big(a^{+}(\varphi),\, a^{-}(\varphi)\big)_{\varphi\in K}$ is irreducible on $\mathcal{F}(K)$, see \cite{jorgensenschmittwernerJFA1995}. Then the Schur Lemma implies that
the family  $\big(a_i^{+}(\varphi),\, a_i^{-}(\varphi)\big)_{\varphi\in K,\, i=1,2}$ is irreducible on $\mathcal{F}$. Evidently,
for $\Omega=\Psi\otimes\Psi$,  we have
$
a_i^{-}(\varphi)\Omega =0$ for all $\varphi\in K$ and $i=1,2$.
Thus, as noted above, the operators defined by (\ref{k_qFock}) determine the Fock representation of the commutation relations \eqref{qds5qsd}.

As a corollary of our description we get the boundedness of the Fock representation of \eqref{qds5qsd}. Indeed, as follows from  \cite{BS},  for  each $\varphi\in K$, the operator $a^{+}(\varphi)$ is bounded and $\|a^{+}(\varphi)\|=\frac{\|f\|}{\sqrt{1-k}}$. Hence,
\[
\|a_i^{+}(\varphi)\|=\frac{\|\varphi\|}{\sqrt{1-k}},\quad i=1,2.
\]
\end{example}

\begin{example}
Consider the following operator $C$ related to the Pusz--Woronowicz twisted canonical commutation relations \cite{pw} (see also \cite{BBB}),
\begin{equation}\label{xesa45w}
C=\left(
\begin{matrix}
\mu^2 & 0 & 0 & 0 \\
0 & 0 & \mu & 0 \\
0 & \mu & \mu^2-1 & 0 \\
0 & 0 & 0 & \mu^2%
\end{matrix}%
\right),\quad \mu\in (0,1).
\end{equation}
We then get
$$\mathbf 1-C^2=\left(\begin{matrix}
1-\mu^4&0&0&0\\
0&1-\mu^2&-\mu(\mu^2-1)&0\\
0&-\mu(\mu^2-1)&1-\mu^2-(\mu^2-1)^2&0\\
0&0&0&1-\mu^4
\end{matrix}\right).$$
An easy calculation then shows that
\begin{align}
\ker(\mathbf 1-C^2)&=\ls\{-\mu e_{12}+e_{21}\},\label{xetsdetse5}\\
\ran(\mathbf 1-C^2)&=\ls\big\{e_{11},e_{22}, e_{12}+\mu e_{21}\big\}\label{xssssres}.
\end{align}
By \eqref{xssssres}, for a function $f^{(2)}\in H^{\otimes 2}$, we have
$$f^{(2)}(x,y)-Cf^{(2)}(y,x)\in \ran(\mathbf 1-C^2)$$
if and only if $\mu\tilde f^{(2)}_{12}=\tilde f^{(2)}_{21}$. Here,
$$\tilde f^{(2)}_{ij}(x,y):=\frac12\big(f^{(2)}_{ij}(x,y)+f^{(2)}_{ij}(y,x)\big).$$
Hence, by Theorem~\ref{vytr7i5}, for $n\ge 2$, $\mathcal F_n(H)$ consists of all functions $f^{(n)}\in H^{\otimes n}$ that satisfy a.e.\ the following symmetry condition:
$$\mu S_k^{(n)} f^{(n)}_{i_1\dots i_{k-1}12i_{k+2}\dots i_n}=S_k^{(n)}f^{(n)}_{i_1\dots i_{k-1}21i_{k+2}\dots i_n} $$
for $k\in\{1,\dots,n-1\}$ and $i_1,\dots,i_{k-1},i_{k+2},\dots,i_n\in\{1,2\}$. Here, for $f^{(n)}\in H^{\otimes n}$ and $i_1,\dots,i_n\in\{1,2\}$,
\begin{align*}
&(S_k^{(n)}f^{(n)}_{i_1,\dots,i_n})(x_1,\dots,x_n)\\
&\quad:= \frac12\big(
f^{(n)}_{i_1,\dots,i_n}(x_1,\dots,x_n)+f^{(n)}_{i_1,\dots,i_n}(x_1,\dots,x_{k-1},x_{k+1},x_{k},x_{k+2},\dots,x_n)
\big).\end{align*}

By using Corollary \ref{rewa5w} and \eqref{xssssres}, we get, for $f^{(2)}\in H^{\otimes 2}$
\begin{align*}
\mathfrak Pf^{(2)}&=f^{(2)}_{11}e_{11}+f^{(2)}_{22}e_{22}+
\big(f^{(2)}_{12}+{\textstyle\frac\mu{1+\mu^2}}(-\mu\tilde f^{(2)}_{12}+\tilde f^{(2)}_{21})\big)e_{12}\\
&\quad+\big(f^{(2)}_{21}-{\textstyle\frac1{1+\mu^2}}(-\mu\tilde f^{(2)}_{12}+\tilde f_{21})\big)e_{21}.
\end{align*}

By \eqref{ctrs6uw} and \eqref{xesa45w},
$$\widetilde C=\left(
\begin{matrix}
\mu^2& 0&0&0\\
0&0&\mu&0\\
0&\mu&0&0\\
\mu^2-1&0&0&\mu^2
\end{matrix}
\right).$$

Hence, by Theorem \ref{yjsqrd6uq} and \eqref{xetsdetse5}, we have, for all $\varphi,\psi\in L^2(X)$,
\begin{align*}
&a_1^-(\varphi)a_1^+(\psi)=\mu^2a_1^+(\psi)a_1^-(\varphi)+\langle\varphi,\psi\rangle,\\
&a_2^-(\varphi)a_2^+(\psi)=(\mu^2-1)a_1^+(\psi)a_1^-(\varphi)
+\mu^2a_2^+(\psi)a_2^-(\varphi)+\langle\varphi,\psi\rangle,\\
&a_1^-(\varphi)a_2^+(\psi)=\mu a_2^+(\psi)a_1^-(\varphi),\\
&a_2^-(\varphi)a_1^+(\psi)=\mu a_1^+(\psi)a_2^-(\varphi),\\
&a_2^+(\varphi)a_1^+(\psi)+a_2^+(\psi)a_1^+(\varphi)=\mu\big(
a_1^+(\varphi)a_2^+(\psi)+a_1^+(\psi)a_2^+(\varphi)\big),\\
&a_1^-(\varphi)a_2^-(\psi)+a_1^-(\psi)a_2^-(\varphi)=\mu\big(
a_2^-(\varphi)a_1^-(\psi)+a_2^-(\psi)a_1^-(\varphi)\big).
\end{align*}
\end{example}

\begin{example}
Consider   the operator $C$ given by the matrix
\begin{equation*}
C=\left(
\begin{matrix}
0 & 0 & 0 & q \\
0 & k & 0 & 0 \\
0 & 0 & k & 0 \\
\bar{q} & 0 & 0 & 0%
\end{matrix}%
\right)
\end{equation*}%
where $q\in\mathbb{C}$, $|q|=1$ and $k\in [-1,1]$.
Then
\begin{equation*}
\mathbf 1-C^{2}=\left(
\begin{matrix}
0 & 0 & 0 & 0 \\
0 & 1-k^{2} & 0 & 0 \\
0 & 0 & 1-k^{2} & 0 \\
0 & 0 & 0 & 0%
\end{matrix}%
\right) .
\end{equation*}%

First assume $|k|<1$. Then
\begin{align}\ran(\mathbf 1-C^{2})&=\ls\{e_{12},e_{21}\},\label{s4wq4yw}\\
 \ker(\mathbf 1-C^{2})&=\ls\{e_{11},e_{22}\}.\label{xw4aq42w7}\end{align}
Just as in Example \ref{ra4tqq}, let $Q(1,2):=\bar q$ and $Q(2,1):=q$. By  Theorem~\ref{vytr7i5} and \eqref{s4wq4yw},  $\mathcal F_n(H)$ consists of all functions $f^{(n)}\in H^{\otimes n}$ that satisfy a.e.\ the following symmetry condition:
\begin{align}
& f^{(n)}_{i_1\dots i_{k-1}i_ki_ki_{k+2}\dots i_n} (x_1,\dots,x_n)\notag\\
&\quad=Q(i_k,j_k)	f^{(n)}_{i_1\dots i_{k-1}j_kj_ki_{k+2}\dots i_n} (x_1,\dots,x_{k-1},x_{k+1},x_k,x_{k+2},\dots,x_n)\label{xerqy453q}\end{align}
for $k\in\{1,\dots,n-1\}$ and $i_1,\dots,i_{k-1},i_k,j_k,i_{k+2},\dots,i_n\in\{1,2\}$, $i_k\ne j_k$,  compare with  \eqref{s4q53y2}.
Similarly  to \eqref{rsw54aq4}, we can easily find the explicit form of $\mathfrak P$.

Furthermore, by \eqref{ctrs6uw},
\begin{equation*}
\widetilde C=\left(
\begin{matrix}
0 & 0 & 0 & k \\
0 & \bar q & 0 & 0 \\
0 & 0 & q & 0 \\
k & 0 & 0 & 0%
\end{matrix}%
\right),
\end{equation*}%
which, by Theorem \ref{yjsqrd6uq} and \eqref{xw4aq42w7}, implies the commutation relations, for any $\varphi,\psi\in L^2(X)$,
\begin{align}
&a^-_1(\varphi)a_1^+(\psi)=ka_2^+(\psi)a^-_2(\varphi)+\langle \varphi,\psi\rangle,\notag\\
&a^-_2(\varphi)a_2^+(\psi)=ka_1^+(\psi)a^-_1(\varphi)+\langle \varphi,\psi\rangle,\notag\\
&a_1^-(\varphi)a_2^+(\psi)=\bar q a_1^+(\psi)a_2^-(\varphi),\notag\\
&a_2^-(\varphi)a_1^+(\psi)=qa_2^+(\psi)a_1^-(\varphi),\notag\\
&a_1^+(\varphi)a_1^+(\psi)=qa_2^+(\psi)a_2^+(\varphi),\notag\\
&a_1^-(\varphi)a_1^-(\psi)=\bar q a_2^-(\psi)a_2^-(\varphi).\label{qydfcfrdy}
\end{align}
Note that the commutation relations \eqref{qydfcfrdy} have a more complex structure than the commutation relations \eqref{qds5qsd}.

In the case $k=\pm1$, the matrix $C$ is unitary, and so $\mathbf 1-C^2=\mathbf 0$.
To describe $\mathcal F_n(H)$, in addition to \eqref{xerqy453q}, the following symmetry condition must be satisfied:
$$f^{(n)}_{i_1\dots i_n}(x_1,\dots,x_n)=k f^{(n)}_{i_1\dots i_{k-1}i_{k+1}i_ki_{k+2}\dots i_n}(x_1,\dots,x_{k-1},x_{k+1},x_{k},x_{k+2},\dots,x_n) $$
 for $k\in\{1,\dots,n-1\}$ and $i_1,\dots,i_n\in\{1,2\}$ with $i_k\ne i_{k+1}$.
Recall also that $\mathbb P_n=\frac1{n!}\mathcal P_n$ in this case. Additionally to the commutation relations \eqref{qydfcfrdy}, it also holds that
\begin{align*}
&a_1^+(\varphi)a_2^+(\psi)=ka_2^+(\psi)a_1^+(\varphi),\\
&a_1^-(\varphi)a_2^-(\psi)=ka_2^-(\psi)a_1^-(\varphi).
\end{align*}

\end{example}

\subsection{Non-Abelian anyon quantum systems}

In this section, we will discuss the case where the operator $C$  depends on spatial variables $(x,y)\in X^{(2)}$ in a special way and determines a non-Abelian anyon quantum system when $d=2$, see \cite{GoMa}.

Recall \eqref{dr6e6ue}. For $x,y\in X^{(2)}$, we will write $x<y$ and $x>y$ if $x^1<y^1$ and $x^1>y^1$, respectively.  Let $C$ be a unitary
operator on $V\otimes V$ and we define $C_{x,y}$ by formula~\eqref{sr4q4w2}.
By \eqref{T-C-map}, we get $T^2=\mathbf1$, hence $T$ is a unitary operator.

\begin{lemma}\label{drs5ywuy}
The operator $T$ satisfies the Yang--Baxter equation \eqref{braid} on $H^{\otimes 3}$  if and only if the operator $C$ satisfies the Yang--Baxter equation \eqref{YBnumer} on $V^{\otimes 3}$.
\end{lemma}

\begin{proof}
Recall Lemma \ref{5e3ss}. In view of \eqref{sr4q4w2}, for $x<y<z$, formula \eqref{YBxyz} becomes   \eqref{YBnumer}.
If $x<z<y$, \eqref{sr4q4w2} obtains the form
\begin{equation}  \label{YBnumerC}
C^{1,2}C^{2,3}(C^{1,2})^* = (C^{2,3})^* C^{1,2} C^{2,3}.
\end{equation}
Multiplying equality \eqref{YBnumerC} by $C^{23}$ from the left and by $C^{12}$ from the right, we arrive at \eqref{YBnumer}. The other remaining cases are similar.
\end{proof}

\begin{remark}

\end{remark}

The next statement is Corollary \ref{ysd6qwe} applied to our case.

\begin{proposition}\label{ctsw6u}
\label{uniFock} For each $n\ge 2$, the space $\mathcal{F}_n (H)$ consists of all functions $f^{(n)}\in H^{\otimes n}$ that satisfy a.e.\ the following symmetry condition:
\begin{equation}  \label{Csymmetry}
f^{(n)}(x_1, \dots , x_n) = C^{k,k+1}_{x_k, x_{k+1}} f(x_1,\dots,x_{k-1}, x_{k+1}, x_k,x_{k+2},\dots , x_n)
\end{equation}
for each $k\in\{1,\dots,n-1\}$.
\end{proposition}

Also recall that, in this case, the orthogonal projection of $H^{\otimes n}$ onto $\mathcal F_n(H)$ satisfies $\mathbb P_n=\frac1{n!}\,\mathcal P_n$.

\begin{example}\label{fytr7r4}
Consider $C$ of the form
\begin{equation}\label{cxeswa}
C=\left(
\begin{matrix}
q_1 & 0 & 0 & 0 \\
0 & 0 & q_3 & 0 \\
0 & q_2 & 0 & 0 \\
0 & 0 & 0 & q_4%
\end{matrix}
\right),
\end{equation}
where $q_1,q_2,q_3,q_4\in\mathbb C$ are of modulus 1.
Define a complex-valued function $Q$ a.e.\ on $(\{1,2\}\times X)^2$ by
\begin{align*}
&Q(1,x,1,y):=\begin{cases}q_1,&\text{if }x<y,\\
\bar q_1,&\text{if }x>y,\end{cases}\qquad Q(2,x,2,y):=\begin{cases}q_4,&\text{if }x<y,\\
\bar q_4,&\text{if }x>y,\end{cases}\\
&Q(1,x,2,y):=\begin{cases}q_3,&\text{if }x<y,\\
\bar q_2,&\text{if }x>y,\end{cases}\qquad Q(2,x,1,y):=\begin{cases}q_2,&\text{if }x<y,\\
\bar q_3,&\text{if }x>y.\end{cases}
\end{align*}
Note that the function $Q$ Hermitian:
$$Q(i,x,j,y)=\overline{Q(j,y,i,y)}.$$
Then, by Proposition \ref{ctsw6u} and \eqref{cxeswa}, for each $n\ge 2$, the space $\mathcal{F}_n (H)$ consists of all functions $f^{(n)}\in H^{\otimes n}$ that satisfy a.e.\ the following symmetry condition:
\begin{align*}
&f^{(n)}_{i_1\dots i_n}(x_1,\dots,x_n)\\
&\quad=Q(i_k,x_k,i_{k+1},x_{k+1})f^{(n)}_{i_1\dots i_{k-1}i_{k+1}i_k i_{k+2}\dots i_n}(x_1,\dots,x_{k-1},x_{k+1},x_k,x_{k+2},\dots,x_n)
\end{align*}
for all $i_1,\dots,i_n\in\{1,2\}$ and $k\in\{1,\dots,n-1\}$.

By \eqref{ctrs6uw} and \eqref{cxeswa}, we get
$C=\widetilde C^T$ and $C^*=\big(\widetilde C^*\big)^T$.
Hence, by Corollary \ref{formalCRU}, we obtain the following formal commutation relations:
\begin{align*}
&a^-_i(x)a^+_j(y)=Q(i,x,j,y)a^+_j(y)a^-_i(x)+\delta(x-y)\delta_{ij},\\
&a^+_i(x)a^+_j(y)=Q(j,y,i,x)a^+_j(y)a^+_i(x),\\
&a^-_i(x)a^-_j(y)=Q(j,y,i,x)a^-_j(y)a^-_i(x).
\end{align*}
\end{example}

\begin{remark}
Note that the commutation relations in Examples~\ref{ra4tqq}  and~\ref{fytr7r4} are governed by a single Hermitian function, $Q(i,j)$ in Example~\ref{ra4tqq} and $Q(i,x,j,y)$ in Example~\ref{fytr7r4}. Therefore, to construct these examples, one could use the theory of commutation relations deformed with a Hermitian, complex-valued function $Q$, whose modulus is bounded by 1, see \cite{BLW-Q}.
\end{remark}

Another example of a non-Abelian anyon quantum system will be discussed below as a special case of Example~\ref{zewq343u}.

\subsection{General spatial dependence}
We will now consider an example of a matrix $C_{x, y}$ with somewhat more complicated dependence on spatial
variables $x, y \in X$.

\begin{example}\label{zewq343u}
Let $Q_1,Q_2:X^{(2)}\to \mathbb C$ satisfy
$$Q_i(x,y)=\overline{Q_i(y,x)},\quad i=1,2,\quad |Q_1(x,y)|\le 1,\quad
|Q_2(x,y)|=1,\qquad  (x,y)\in X^2.$$

Let matrix $C_{x, y}$  have the form
\begin{equation*}
C_{x, y}=\left(
\begin{matrix}
0 & 0 & 0 & Q_1(x,y) \\
0 &  Q_2(x,y) & 0 & 0 \\
0 &  0 & Q_2(x,y)  & 0 \\
Q_1(x,y) & 0 & 0 & 0%
\end{matrix}%
\right).
\end{equation*}
Note that $C_{x,y}=C_{y,x}^*$.
A direct calculation shows that  $C_{x, y}$
satisfies the Yang--Baxter equation (\ref{YBxyz}).
For  $x, y\in X^{(2)}$, we have
$$\mathbf 1 - C_{x, y}C^*_{x, y}= \left(
\begin{matrix}
1-|Q_1(x,y)|^2 & 0 & 0 & 0 \\
0 &  0 & 0 & 0 \\
0 &  0 &0  & 0 \\
0 & 0 & 0 & 1-|Q_1(x,y)|^2%
\end{matrix}%
\right).$$
We denote
$$Y:=\big\{(x,y)\in X^2\mid Q_1(x,y)|^2=1\},\quad Z:=\{(x,y)\in X^2\mid Q_1(x,y)|^2<1\big\}.$$
Then, for $(x,y)\in Y$, $\mathbf 1 - C_{x, y}C^*_{x, y}=\mathbf 0$, and for all $(x,y)\in Z$,
$$\ker (\mathbf 1 - C_{x, y}C^*_{x, y})=\ls\{e_{12},e_{21}\},\quad \ran(\mathbf 1 - C_{x, y}C^*_{x, y})=\ls\{e_{11},e_{22}\}. $$
Hence, by Theorem~\ref{vytr7i5}, for $n\ge 2$, $\mathcal F_n(H)$ consists of all functions $f^{(n)}\in H^{\otimes n}$ that satisfy a.e.\ the following symmetry conditions:
\begin{align*}
&f^{(n)}_{i_1\dots i_{k-1}11i_{k+2}\dots i_n}(x_1,\dots,x_n)\\
&\quad =
Q_1(x_k,x_{k+1})f^{(n)}_{i_1\dots i_{k-1}22i_{k+2}\dots i_n}
(x_1,\dots,x_{k-1},x_{k+1},x_k,x_{k+2},\dots,x_n)\quad\text{if }(x_k,x_{k+1})\in Y,\\
 &f^{(n)}_{i_1\dots i_{k-1}12i_{k+2}\dots i_n}(x_1,\dots,x_n)\\
&\quad =
Q_2(x_k,x_{k+1})f^{(n)}_{i_1\dots i_{k-1}12i_{k+2}\dots i_n}
(x_1,\dots,x_{k-1},x_{k+1},x_k,x_{k+2},\dots,x_n),\\
&f^{(n)}_{i_1\dots i_{k-1}21i_{k+2}\dots i_n}(x_1,\dots,x_n)\\
&\quad =
Q_2(x_k,x_{k+1})f^{(n)}_{i_1\dots i_{k-1}21i_{k+2}\dots i_n}
(x_1,\dots,x_{k-1},x_{k+1},x_k,x_{k+2},\dots,x_n),
\end{align*}
for all $i_1,\dots,i_{k-1},i_{k+2},\dots,i_n\in\{1,2\}$ and $k\in\{1,\dots,n-1\}$.

In the case where the set $Z$ is empty (or of zero measure), the corresponding operator $T$ is unitary, hence $\mathbb P_n=\frac1{n!}\mathcal P_n$. If the set $Z$ is of positive measure, the form $(\mathfrak Pf^{(2)})(x,y)$ will depend on whether $(x,y)$ is a point of $Y$ or $Z$. In both cases, the explicit form of $(\mathfrak Pf^{(2)})(x,y)$ can be easily calculated by using Corollary \ref{rewa5w}. We leave the details to the interested reader.

By \eqref{ctrs6uw}, we get
\begin{equation*}
\widetilde C_{x, y}=\left(
\begin{matrix}
0 & 0 & 0 & Q_2(x,y) \\
0 &  Q_1(x,y) & 0 & 0 \\
0 &  0 & Q_1(x,y)  & 0 \\
Q_2(x,y) & 0 & 0 & 0%
\end{matrix}%
\right).
\end{equation*}
Hence, by Corollary \ref{formalCR}, we get the following formal commutation  relations:
\begin{align}
&a_1^-(x)a_1^+(y)=Q_2(x,y)a_2^+(y)a_2^-(x)+\delta(x-y),\notag\\
&a_2^-(x)a_2^+(y)=Q_2(x,y)a_1^+(y)a_1^-(x)+\delta(x-y),\notag\\
&a_1^-(x)a_2^+(y)=Q_1(x,y)a_1^+(y)a_2^-(x),\notag\\
&a_2^-(x)a_1^+(y)=Q_1(x,y)a_2^+(y)a_1^-(x),\notag\\
&a_1^+(x)a_2^+(y)=Q_2(x,y)a_1^+(y)a_2^+(x),\notag\\
&a_2^+(x)a_1^+(y)=Q_2(x,y)a_2^+(y)a_1^+(x),\notag\\
&a_1^+(x)a_1^+(y)=Q_1(x,y)a_2^+(y)a_2^+(x)\quad\text{if }(x,y)\in Y,\notag\\
&a_1^-(x)a_2^-(y)=Q_2(x,y)a_1^-(y)a_2^-(x),\notag\\
&a_2^-(x)a_1^-(y)=Q_2(x,y)a_2^-(y)a_1^-(x),\notag\\
&a_1^-(x)a_1^-(y)=Q_1(x,y)a_2^-(y)a_2^-(x)\quad\text{if }(x,y)\in Y.\label{tsw5w}\end{align}

Let us consider a special case of such a construction. Fix any $q_1,q_2\in\mathbb C$ with $|q_1|=|q_2|=1$ and define
$$Q_i(x,y)=\begin{cases}
q_i,&\text{if }x<y,\\ \bar q_i,&\text{if }x>y,\end{cases}\quad i=1,2.$$
With such a choice of functions $Q_1$, $Q_2$ and $d=2$, the above construction gives an example of a non-Abelian anyon quantum system with the operator
$$C= \left(
\begin{matrix}
0 & 0 & 0 & q_1 \\
0 &  q_2 & 0 & 0 \\
0 &  0 & q_2  & 0 \\
q_1 & 0 & 0 & 0%
\end{matrix}%
\right).$$
Note that, in this case, the commutation relations \eqref{tsw5w} hold for all $(x,y)\in X^{(2)}$.

Further examples of such a construction can be achieved by choosing
$$Q_1(x,y)=ke^{i\alpha(x-y)},\quad Q_2(x,y)=e^{i\beta(x-y)},$$
where $k\in[-1,1]$ and $\alpha,\beta\in\mathbb R$.

\end{example}

{\bf Acknowledgements}. AD, EL and DP are grateful to the London Mathematical Society for partially supporting the visit of DP to Swansea University and University of York. The authors are grateful to Marek Bo\.zejko, Ivan Feshchenko, Gerald Goldin, Alexey Kuzmin and Janusz Wysocza\'nski for useful discussions.

\end{document}